\documentclass[11pt]{article}

\pdfoutput = 1

\usepackage{amsmath,amsfonts,amssymb}

\usepackage{geometry}

\usepackage{amsxtra}

\usepackage{array,dcolumn}
\usepackage{graphicx} \usepackage{color}

\usepackage[hyperfootnotes=true,
        pdffitwindow=true,
        plainpages=false,
        pdfpagelabels=true,
        pdfpagemode=UseOutlines,
        pdfpagelayout=SinglePage,
        hyperindex,]{hyperref}

\numberwithin{equation}{section} 

\usepackage{amsthm}

\theoremstyle{plain}
\newtheorem{theorem}{Theorem}[section]
\newtheorem{lemma}[theorem]{Lemma}
\newtheorem{corollary}[theorem]{Corollary}
\newtheorem{proposition}[theorem]{Proposition}

\theoremstyle{definition}

\newtheorem{remark}[theorem]{Remark}

\newcommand{\R}{\mathbb R}
\newcommand{\C}{\mathbb C}

\newcommand{\even}{\textup{even}}
\newcommand{\odd}{\textup{odd}}

\DeclareMathOperator{\Tr}{Tr} 

\DeclareMathOperator{\gO}{O}
\DeclareMathOperator{\gU}{U}
\DeclareMathOperator{\gUSp}{USp}

\newcommand{\MeijerG}[8][\bigg]{G^{{ #2 },{ #3 }}_{{ #4 },{ #5 }} #1( \begin{matrix} #6 \\ #7 \end{matrix}\, #1\vert\, #8 #1)}
\newcommand{\hypergeometric}[6][\bigg]{\,{}_{#2} F_{#3} #1( \begin{matrix} #4 \\ #5 \end{matrix}\, #1\vert\, #6 #1)}

\begin{document}

\title{\bfseries\Large Orthogonal and symplectic Harish-Chandra integrals \\ and matrix product ensembles}

\author{
Peter J.~Forrester\footnotemark[1], ~Jesper R. Ipsen\footnotemark[1], ~Dang-Zheng Liu\footnotemark[2] ~and Lun Zhang\footnotemark[3]
}

\renewcommand{\thefootnote}{\fnsymbol{footnote}}
\footnotetext[1]{Department of Mathematics and Statistics, ARC Centre of Excellence for Mathematical and Statistical Frontiers,
The University of Melbourne, Victoria 3010, Australia. E-mail: \{pjforr, jesper.ipsen\}@unimelb.edu.au}
\footnotetext[2]{Key Laboratory of Wu Wen-Tsun Mathematics, CAS, School of Mathematical Sciences, University of Science and Technology of China, Hefei 230026, P.R.~China. E-mail: dzliu@ustc.edu.cn}
\footnotetext[3] {School of Mathematical Sciences and Shanghai Key Laboratory for Contemporary Applied Mathematics, Fudan University, Shanghai 200433, P. R. China. E-mail: lunzhang@fudan.edu.cn}
\maketitle

\begin{abstract}
In this paper, we highlight the role played by orthogonal and symplectic Harish-Chandra integrals in the study of real-valued matrix product ensembles. By making use of these integrals and the matrix-valued Fourier-Laplace transform, we find the explicit eigenvalue distributions for  particular Hermitian anti-symmetric matrices and particular Hermitian anti-self dual matrices, involving both sums and products.
 As a consequence of these results, the eigenvalue probability density function of the random product structure $X_M \cdots  X_1( iA) X_1^T  \cdots X_M^T$, where each
$X_i$ is a standard real   Gaussian matrix, and $A$ is a real anti-symmetric matrix can be determined. For $M=1$ and $A$ the bidiagonal anti-symmetric matrix with 1's above the diagonal, this reclaims results of
Defosseux. For general $M$, and this choice of $A$, or $A$ itself a standard Gaussian anti-symmetric
matrix, the eigenvalue distribution is shown to coincide with that of the
  squared singular values
 for the product of certain complex Gaussian matrices first studied by Akemann et al. As
 a point of independent interest, we also include a self-contained diffusion equation derivation of the orthogonal and symplectic Harish-Chandra integrals.
\end{abstract}


\section{Introduction}
Let $A$ and $B$ be $N\times N$ Hermitian matrices with eigenvalues $\{a_j\}_{j=1}^N$ and $\{b_j\}_{j=1}^N$, respectively. For $U\in \gU(N)$, the matrix group of $N\times N$ complex unitary matrices, let $(U^\dagger dU)$ denote the invariant measure \cite{Hu97,DF17} normalised so that
$$\int (U^\dagger dU)=1,$$
where the superscript $^\dagger$ stands for the conjugate transpose. According to a result due to Harish-Chandra \cite{HC57} and Itzykson and Zuber \cite{IZ80}, we have
\begin{equation}\label{1}
\int_{\gU(N)} e^{\Tr(UAU^\dagger B)}\,(U^\dagger dU)=
\prod_{j=1}^{N}\Gamma(j)\,\frac{\det[e^{a_j b_k}]_{j,k=1}^N}{\Delta_N(a)\Delta_N(b)},
\end{equation}
where
\begin{equation}
\Delta_N(x)=\det[x_j^{k-1}]_{j,k=1}^N=\prod_{1\leq j<k\leq N} (x_k-x_j)
\end{equation}
is the Vandermonde determinant. This matrix integral formula is fundamental to the study of random matrices of the form $$H_0+tH,$$
where $H_0$, $H$ are Hermitian with $H_0$ fixed and $H$ a member of the GUE (Gaussian Unitary Ensemble); see e.g. \cite{DF06b}. It also underlies exact calculations relating to the class of
Wishart matrices $X^\dagger \Sigma X$, where $X$ is an $m \times n$ standard complex
Gaussian matrix and $\Sigma$ a fixed positive definite matrix \cite{BBP05}.

The result \eqref{1} is well known in random matrix theory. It turns out that from the general group integral evaluated in Harish-Chandra's paper other lesser known integration formulas of the type \eqref{1} can be inferred; see e.g. \cite{Ey07}. To state these, let $\gO(N)$ denote the matrix group of $N\times N$ real orthogonal matrices and let $\gUSp(2N)$ denote the matrix group of $2N\times 2N$ unitary symplectic matrices. For $R\in \gO(N)$ and $S\in \gUSp(2N)$, let $(R^T dR)$ and $(S^\dagger dS)$ denote the corresponding invariant measures with the superscript $^T$ being transpose, normalised to integrate to unit. Elements of the latter are elements of $\gU(2N)$, with the additional symmetry
$$UZ_{2N}U^T=Z_{2N},$$
where $Z_{2N}=\mathbb{I}_N\otimes\big[\begin{smallmatrix}0 &  1 \\ -1 & 0 \end{smallmatrix}\big]$.

Let $X$ and $Y$ be $N\times N$ anti-symmetric real matrices with non-zero eigenvalues $\{\pm ix_j\}_{j=1}^{[N/2]}$ and $\{\pm iy_j\}_{j=1}^{[N/2]}$ ($x_j,y_j>0$), respectively, where $[x]$ denotes the integer part of $x$. With $N=2m$ even, one has
\begin{equation}\label{2}
\int_{\gO(2m)} e^{\frac{1}{2}\Tr(XRYR^T)}\,(R^T dR)=c_m^\even \frac{\det[2\cosh (x_iy_j)]_{i,j=1}^m}{\Delta^\even _m(x)\Delta^\even _m(y)},
\end{equation}
while for $N=2m+1$ odd,
\begin{equation}\label{3}
\int_{\gO(2m+1)} e^{\frac{1}{2}\Tr(XRYR^T)}\,(R^T dR)=c_m^\odd \frac{\det[2\sinh (x_iy_j)]_{i,j=1}^m}{\Delta^\odd _m(x)\Delta^\odd _m(y)},
\end{equation}
with proportionality constants
\begin{equation}
c_m^\even =\prod_{j=1}^{m}\frac{\Gamma(2j-1)}{2}
\qquad\text{and}\qquad
c_m^\odd =\prod_{j=1}^{m}\frac{\Gamma(2j)}{2}.
\end{equation}
For notational simplicity, we have introduced even/odd modifications of the Vandermonde determinant defined by
\begin{equation}\label{4}
\Delta^\even _m(u)=\prod_{1\leq j<k\leq m}(u_k^2-u_j^2)
\qquad \text{and}\qquad
\Delta^\odd _m(u)=\prod_{j=1}^m u_j\prod_{1\leq j<k\leq m}(u_k^2-u_j^2),
\end{equation}
respectively.

Also, in the case that $X$ and $Y$ are $2N\times 2N$ anti-Hermitian matrices (i.e., $X^\dagger=-X$, $Y^\dagger=-Y$) with $2\times 2$ blocks of the form
\begin{equation}\label{q}
\begin{bmatrix}
z & w \\
-\overline{w} & \overline{z}
\end{bmatrix}, \quad z,w\in \C,
\end{equation}
--- this being the standard $2\times2$ matrix representation of a quaternion --- one has
\begin{equation}\label{5}
\int_{\gUSp(2N)} e^{\frac{1}{2}\Tr(XSY^\dagger S^\dagger)}\,(S^\dagger dS)=c_N^\odd \frac{\det[2\sinh (x_iy_j)]_{i,j=1}^N}{\Delta^\odd _N(x)\Delta^\odd _N(y)},
\end{equation}
which is identical to \eqref{3} with $m=N$.

In addition to the fundamental role played by the Harish-Chandra/ Itzkykson--Zuber integral
(\ref{1}) in the study of sums involving random matrices, recent findings have uncovered that this
same matrix integral also underpins the explicit calculation of the joint probability density function (PDF) for singular values of products of complex Gaussian random matrices \cite{AKW13,AIK13}. While it is well known that the
orthogonal and symplectic Harish-Chandra integrals relate to sums involving random matrices (see e.g.~\cite[Cor.~12]{KT04}), the role they play in relation to products is not part of the existing literature. However, there is a hint from the work of Defosseux \cite{De10}. There the random product matrices
\begin{equation}\label{5a}
X \Big ( \mathbb I_m \otimes \begin{bmatrix} 0 & i \\ -i & 0 \end{bmatrix} \Big ) X^T, \qquad
X \Big ( \Big ( \mathbb I_m \otimes \begin{bmatrix} 0 & i \\ -i & 0 \end{bmatrix} \Big )
\oplus [0] \Big ) X^T,
\end{equation}
with $X$ a real Gaussian matrix of sizes $2n \times 2m$ and $(2n+1) \times (2m+1)$
respectively, as well the random product
matrix
\begin{equation}\label{5b}
X \Big ( \mathbb I_m \otimes  \begin{bmatrix} 1 & 0 \\ 0 & -1 \end{bmatrix} \Big ) X^\dagger
\end{equation}
with $X$ a $2n \times 2m$ complex Gaussian matrix with each $2 \times 2$ block of the form
(\ref{q}) and thus representing a quaternion, were introduced and analysed. They were shown to be
closely related to the representation theory associated with the classical groups
${\rm O}(2m)$, ${\rm O}(2m+1)$ and ${\rm USp}(m)$ respectively. In the text below the proof of
Th.~5.3 in  \cite{De10}, it is said that `One easy way to compute the law of the eigenvalues of $X X^\dagger$ for $X$ an $n \times m$ complex Gaussian matrix is to use the Harish-Chandra/ Itzkykson--Zuber integral. But this method does not work for the other fields'. We will show that the Harish-Chandra matrix integrals for the orthogonal and symplectic groups do in fact allow the law of the eigenvalues of the random matrices \eqref{5a} and \eqref{5b} to be computed.

The even and odd cases of the orthogonal group, and the unitary symplectic group require separate treatment, although the working is very similar. For the even dimensional case of the orthogonal group, the eigenvalue PDF of the corresponding product matrix in (\ref{5a}) follows as a result of the following more general result, to be derived using (\ref{2}) in \S \ref{sec:3.2} below.

\begin{theorem}\label{prop4}
Let $\Omega$ be a $2n\times 2n$ Haar distributed real orthogonal matrix, $A$ and $B$ be Hermitian anti-symmetric matrices of size $2n\times 2n$ and $l \times l$, let $X$ be a $2n \times l$ real standard Gaussian matrix (i.e., each entry is a standard Gaussian),  and let $Y$ be a $2n\times 2n$ standard Gaussian Hermitian anti-symmetric matrix (i.e., each entry above the diagonal is $i$ times a standard real Gaussian, joint PDF is proportional to
$\exp(-\frac{1}{4} {\rm Tr} \, Y^2))$.
Define the  $2n\times 2n$ Hermitian anti-symmetric matrix
\begin{equation}\label{3.4}
M =\Omega A\Omega^T+XBX^T+\sqrt{t}Y
\end{equation}
with $t>0$ be a positive parameter.

Let the positive eigenvalues of $A$ and $M$ be denoted $\{a_j\}_{j=1}^n$ and $\{\lambda_j\}_{j=1}^n$; assume that they are distinct. The eigenvalue PDF of $M $ is equal to
\begin{equation}\label{3.5}
\frac{{\Delta}^\even_n (\lambda)}{{\Delta}^\even_n (a)}\det[g_k(\lambda_j)]_{j,k=1}^n, \qquad \lambda_1<\lambda_2<\cdots <\lambda_n,
\end{equation}
where ${\Delta}^\even_n (u)$ is given in \eqref{4} and
\begin{equation}\label{3.6}
g_k(\lambda)=\frac{2}{\pi}\int_0^\infty e^{-\frac{1}{2}tc^2}\frac{\cos (a_k c)\,\cos(\lambda c)}{\det(\mathbb{I}+icB)}\,dc.
\end{equation}
\end{theorem}

It will be shown in \S \ref{sec:ele} that, as a consequence of this result, the eigenvalue PDF of
a more general random product structure can be determined.

\begin{corollary}\label{cor1even}
For $j=1,\dots,M$, let $X_j$ be a real standard Gaussian matrix of size
$2 (m + \nu_j) \times 2 (m + \nu_{j-1})$ with $\nu_j \ge \nu_{j-1}$ and $\nu_0 = 0$.
Define the random product matrix
\begin{equation}\label{eq:newproduct}
X_M \cdots X_1  \Big( \mathbb I_m \otimes \begin{bmatrix} 0 & i \\ -i & 0 \end{bmatrix} \Big )
X_1^T \cdots X_M^T.
\end{equation}
With $\{ \lambda_j \}_{j=1}^m$ denoting the positive eigenvalues, the variables
$x_j = \lambda_j^2/2^{2M}$ are distributed as for the eigenvalues of the product ensemble
\begin{equation}\label{product-even}
G_{2M}^\dagger\cdots G_1^\dagger G_1\cdots G_{2M},
\end{equation}
where $\{G_i\}$ are $m\times m$ complex random matrices with PDF proportional to
\begin{equation}\label{4.33}
\det (G_{2i}^\dagger G_{2i})^{\nu_{2i}-1/2} e^{-\Tr G_{2i}^\dagger G_{2i}}
\qquad \text{and} \qquad
\det (G_{2i-1}^\dagger G_{2i-1})^{\nu_{2i-1}} e^{-\Tr G_{2i-1}^\dagger G_{2i-1}}.
\end{equation}
\end{corollary}
The significance of this result is that the exact functional form of the eigenvalue PDF of (\ref{product-even}) is known from \cite{AIK13}. We note that the ordering of this matrix product is irrelevant to the singular value distribution due to the weak commutation relation established in~\cite{IK14}. As a consequence of this exact link to a known model, the correlation kernel can be obtained directly from the existing literature. Likewise, the scaling behaviour of the kernel near the origin~\cite{KZ14} as well as in the bulk and at the soft edge~\cite{LWZ16}. It is also worthwhile to mention that our new product~\eqref{eq:newproduct} gives the possibility to construct complex product matrices with half-integer indices $\{\nu_i\}$
(albeit interlaced with integer indices), while the rectangular matrices studied by Akemann et al.~\cite{AIK13} only give rise to integer indices. 

As aforementioned, the odd case of the orthogonal group and the unitary symplectic group, require
separate treatments. It turns out that, as for the Harish-Chandra group integrals, the result is the same
in both cases.

The rest of this paper is organised as follows. In \S \ref{sec2} we give a self-contained derivation of (\ref{1}), making use of a characterisation as the solution of a certain diffusion equation, which might be of independent interest. In \S \ref{sec3} we derive Theorem \ref{prop4}, and its analogue in the odd dimensional case, as well as for the matrix structure relating to (\ref{3.4}) in the case of the unitary symplectic group. In the final subsection we give a technically different derivation of Theorem \ref{prop4}, which in keeping with the recent works \cite{KK16a,KK16b,KR16} highlights the role of matrix transforms, although again the matrix integral (\ref{2}) plays a key role. Special cases, including the result Corollary \ref{cor1even} are given in \S \ref{sec4} for the even dimensional case, and in \S \ref{sec5} for the odd dimensional case.

\section{Fokker--Planck equations relating to the matrix integrals}\label{sec2}

Beyond the abstract working in Harish-Chandra's original paper, or an appeal to the theory of Duistermatt--Heckman localisation \cite{DH82} as outlined in \cite{Ey07}, Guhr and Kohler \cite{GK04} have pointed out that the matrix integrals \eqref{2}--\eqref{4} can be derived using the diffusion equation approach analogous to that used by Itzykson and Zuber \cite{IZ80} to deduce \eqref{1}. However, their working and final result is in the supersymmetric setting. In reading their result, one is thus faced with the added complexity of at least an implicit assumption  of familiarity with the theory of supergroups, and use of theorems within that theory. It thus remains to give a self-contained diffusion equation derivation of the matrix integrals \eqref{2}--\eqref{4}. Here we take up this problem, giving the full details in the case of \eqref{2}.

\begin{proposition}\label{prop1}
Let $X$ and $Y$ be $2m\times 2m$ real anti-symmetric matrices as in \eqref{2}. For $t>0$ define
\begin{equation}\label{2.1}
P_t^\even (X|Y)=\left(\frac{1}{2\pi t}\right)^{m(2m-1)/2}e^{-\Tr(X-Y)^2/4t}
\end{equation}
and
\begin{equation}\label{2.2}
p_t^\even (x|y)
=\Delta^\even _m(x)^2\left(\frac{1}{2\pi}\right)^m {\rm vol} \, {\rm O}(2m) \,
\int_{\gO(2m)} P_t^\even (R\Lambda_X R^T|Y)\,(R^T dR),
\end{equation}
with $\{\pm ix_j\}_{j=1}^{m}$ denoting the eigenvalues of $X$ and
\begin{equation}\label{2.3}
\Lambda_X=\operatorname{diag}\Big(
\big[\begin{smallmatrix} 0 & x_1 \\ -x_1 & 0 \end{smallmatrix}\big]
,\ldots,
\big[\begin{smallmatrix} 0 & x_m \\ -x_m & 0 \end{smallmatrix}\big]
\Big).
\end{equation}
Then the quantity $p_t^\even $, which corresponds to the eigenvalue PDF of $X$, satisfies the Fokker--Planck equation
\begin{equation}\label{2.4}
\mathcal{L}\,p_t^\even =\frac{\partial}{\partial t}p_t^\even ,
\qquad
\mathcal{L}:=\sum_{j=1}^m \frac{\partial}{\partial x_j}\left(\frac{\partial W}{\partial x_j}
+\frac{1}{2}\frac{\partial}{\partial x_j}\right)
\end{equation}
with
\begin{equation}\label{2.5}
W=-\sum_{1\leq j<k\leq m}\log\lvert x_j^2-x_k^2\rvert.
\end{equation}
\end{proposition}
\begin{proof}
We view \eqref{2.1} as the eigenvalue PDF for an ensemble of random anti-symmetric matrices $ {X }$ given a fixed anti-symmetric matrix $Y$. Noting that
\begin{equation*}
P_t^\even (X|Y)=\prod_{1\leq j<k\leq 2m}\left(\frac{1}{2\pi t}\right)^{1/2} e^{-(X_{jk}-Y_{jk})^2/2t},
\end{equation*}
where $X=[X_{jk}]_{j,k=1}^{2m}$, $Y=[Y_{jk}]_{j,k=1}^{2m}$, it follows from the fact that the functional form in the product satisfies a one-dimensional diffusion (heat) equation, that $P_t^\even$ itself satisfies the multi-dimensional diffusion equation
\begin{equation}\label{2.5a}
\frac{1}{2}\sum_{1\leq j<k\leq 2m}\frac{\partial^2}{\partial X_{jk}^2}P_t^\even (X|Y)=\frac{\partial}{\partial t}P_t^\even (X|Y).
\end{equation}

The anti-symmetric matrices $X$ can be put into block diagonal form \eqref{2.3} by conjugating with a real orthogonal matrix
\begin{equation}\label{2.6}
X=R\Lambda_X R^T.
\end{equation}
In terms of the variables implied by \eqref{2.6} --- the eigenvalues and eigenvectors of $X$ --- the measure $(dX)$ decomposes (see e.g. \cite[Exercises 1.3 q.5(iii)]{Fo10})
\begin{equation}\label{2.6a}
(dX)=\prod_{1\leq j<k\leq 2m}dX_{jk}=\Delta_m^\even (x)^2 {\rm vol} \, {\rm O}(2m) \prod_{j=1}^mdx_j\,(R^TdR).
\end{equation}
For this to be a bijection, we require that the $x_j$'s be ordered and that $R\in \gO(2m)/\gO(2)^m$. On the other hand, since the block-diagonal subgroup $\gO(2)^m$ commute with $\Lambda_X$, we may extend the integration to $R\in \gO(2m)$ provided we divide \eqref{2.6a} by $\operatorname{vol} \gO(2)^m=(2\pi)^m$. Now multiplying by \eqref{2.1}, and integrating over $R\in \gO(2m)$, it follows that $p_t^\even $ as defined by \eqref{2.2} is the eigenvalue PDF of $X$. Next change variables $$R\mapsto R_0R$$ for $R_0$ such that $$R_0YR_0^T=\Lambda_Y,$$ where $\Lambda_Y$ is the block diagonal matrix corresponding to the eigenvalues of $Y$ as in \eqref{2.3}. Since $(R^TdR)$ is invariant under this mapping, it follows that \eqref{2.2} only depends on the eigenvalues of $Y$, as indicated in the notation for $p_t$. Finally, change variables $$R\mapsto RR_1$$ and redefine $X$ as $X=R_1\Lambda_X R_1^T$ to conclude
\begin{equation}\label{2.7}
p_t^\even (x|y)
=\Delta^\even _m(x)^2
\left(\frac{1}{2\pi}\right)^m {\rm vol} \, {\rm O}(2m)\int_{\gO(2m)} P_t^\even (RXR^T|\Lambda_Y)\,(R^TdR).
\end{equation}
In view of \eqref{2.7}, the fact that $P_t^\even (X|Y)$ satisfies the diffusion equation \eqref{2.5a} tells us that so does
\begin{equation}\label{2.7a}
\frac{p_t^\even (x|y)}{\Delta^\even _m(x)^2}.
\end{equation}

There is a well established theory to change variables from the Laplacian in terms of the co-ordinates of the independent entries of the matrix, and the co-ordinates of the eigenvalues and eigenvectors as given by \eqref{2.6}; see e.g. \cite[working leading to (11.13)]{Fo10}. Application of this theory gives
\begin{equation*}
\sum_{1\leq j<k\leq 2m}\frac{\partial^2}{\partial X_{jk}^2}
=\frac{1}{\Delta^\even _m(x)^2}\sum_{j=1}^m \frac{\partial}{\partial x_j}
\left(\Delta^\even _m(x)^2\frac{\partial}{\partial x_j}\right)+\mathcal{O}_R,
\end{equation*}
where $\mathcal{O}_R$ denotes an operator with respect to variables relating to $R$. Thus
\begin{equation*}
\frac{1}{2}\sum_{j=1}^m \frac{\partial}{\partial x_j}\left(\Delta^\even _m(x)^2\frac{\partial}{\partial x_j}\right)
\frac{p_t^\even }{\Delta^\even _m(x)^2}=\frac{\partial}{\partial t}p_t^\even .
\end{equation*}
Applying the product rule on the left-hand side gives \eqref{2.4}.
\end{proof}

Generally the Fokker--Planck operator $\mathcal{L}$ in \eqref{2.4} has the property that
\begin{equation}\label{2.8}
e^W\mathcal{L}e^{-W}=\frac{1}{2}\sum_{j=1}^m\left(\frac{\partial^2}{\partial x_j^2}-\left(\frac{\partial W}{\partial x_j}\right)^2+\frac{\partial^2 W}{\partial x_j^2}\right),
\end{equation}
and is thus a symmetric operator. In the particular case that $W$ is given by \eqref{2.5}, the right-hand side of \eqref{2.8} simplifies significantly.

\begin{lemma}
Let $W$ be given by \eqref{2.5}. We have
\begin{equation}\label{2.9}
\sum_{j=1}^m\left(\frac{\partial^2 W}{\partial x_j^2}-\left(\frac{\partial W}{\partial x_j}\right)^2\right)=0.
\end{equation}
\end{lemma}
\begin{proof}
It follows from the definition of $W$ that
\begin{equation*}
\sum_{j=1}^m\left(\frac{\partial W}{\partial x_j}\right)^2
=\sum_{\substack{k_1,k_2=1 \\k_1,k_2\neq j}}^m \frac{4x_j^2}{(x_j^2-x_{k_1}^2)(x_j^2-x_{k_2}^2)}
\end{equation*}
and
\begin{equation*}
\sum_{j=1}^m\frac{\partial^2 W}{\partial x_j^2}=\sum_{\substack{j,k=1\\k\neq j}}^m\frac{4x_j^2}{(x_{j}^2-x_{k}^2)^2}.
\end{equation*}
Separating off the term $k_1=k_2$ in the formula, it follows that
\begin{equation*}
\sum_{j=1}^m\left(\frac{\partial^2 W}{\partial x_j^2}-\left(\frac{\partial W}{\partial x_j}\right)^2\right)
=4\sum_{\substack{j,k_1,k_2=1\\k_1\neq k_2\neq j}}^m\frac{x_j^2}{(x_j^2-x_{k_1}^2)(x_j^2-x_{k_2}^2)}.
\end{equation*}
Moreover, for $a,b,c$ pairwise distinct, we have the algebraic identity
\begin{equation*}
\frac{a}{(a-b)(a-c)}+\frac{b}{(b-a)(b-c)}+\frac{c}{(c-a)(c-b)}=0,
\end{equation*}
and \eqref{2.9} follows.
\end{proof}

\begin{corollary}\label{cor1}
Let $0<x_1<...<x_m$ and $0<y_1<...<y_m$.
We have
\begin{equation}\label{2.10}
\frac{1}{2}\left(\sum_{j=1}^m\frac{\partial^2}{\partial x_j^2}\right)\frac{p_t^\even }{|\Delta^\even _m(x)|}
=\frac{\partial}{\partial t}\frac{p_t^\even }{|\Delta^\even _m(x)|}
\end{equation}
with initial condition
\begin{equation}\label{2.12}
\left\lvert\frac{\Delta^\even _m(y)}{\Delta^\even _m(x)}\right\rvert p_t^\even (x|y)\xrightarrow{t\to 0^+}
 \prod_{i=1}^m\delta(x_i-y_i)
\end{equation}
and the additional requirement that $p_t^\even(x,y)/\Delta^\even _m(x)$ is anti-symmetric in $\{x_i\}_{i=1}^m$ and even in each $x_i$.
\end{corollary}
\begin{proof}
It follows from \eqref{2.4} that
\begin{equation*}
(e^W\mathcal{L}e^{-W})(e^Wp_t^\even )=\frac{\partial}{\partial t}(e^Wp_t^\even ).
\end{equation*}
Making use of \eqref{2.8} with the simplification \eqref{2.9} gives \eqref{2.10}.

It follows from \eqref{2.1} that as $t\to 0^+$ we must have $X\to Y$ and thus with the ordering of eigenvalues
\begin{equation*}
p_t^\even (x|y)\to   \prod_{i=1}^m\delta(x_i-y_i),
\end{equation*}
which is equivalent to \eqref{2.12}. So with the ordering $0<x_1<...<x_m$, we see from \eqref{2.2} that $p_t^\even/\Delta^\even _m(x)$ is equal to $\Delta_m^\even (x)$ times a symmetric function in $\{x_i\}$, and is thus anti-symmetric. We see too that this quantity is even in each $x_i$.
\end{proof}

One viewpoint on \eqref{2.10} is as an imaginary time Schr\"odinger equation with anti-symmetric wave function (in both $\{x_i\}$, $\{y_i\}$)
\begin{equation}\label{2.13}
\psi_t^\even (x|y)=\left\lvert\frac{\Delta_m^\even (y)}{\Delta_m^\even (x)}\right\rvert p_t^\even (x|y),
\end{equation}
corresponding to spinless free fermions. The even Green function solution, i.e. solution satisfying the delta function initial condition on the right-hand side of \eqref{2.12}, is obtained simply by anti-symmetrising the one-dimensional solution.

\begin{proposition}\label{prop2}
Define $\psi_t^\even $ by \eqref{2.13}. We have
\begin{equation}\label{2.14}
\psi_t^\even (x|y)
=\det\bigg[\frac{e^{-(x_j-y_k)^2/(2t)}+e^{-(x_j+y_k)^2/(2t)}}{\sqrt{2\pi t}}\bigg]_{j,k=1}^m.
\end{equation}
\end{proposition}
\begin{proof}
The even solution of the one-particle imaginary time Schr\"odinger equation
\begin{equation*}
\frac{1}{2}\frac{\partial^2}{\partial x^2}\psi_t^\even (x|y)=\frac{\partial}{\partial t}\psi_t^\even (x|y),
\end{equation*}
subject to the initial condition $ \psi_t(x|y)\to\delta(x-y)$  as $t\to 0^+$ is
\begin{equation*}
\psi_t^\even (x|y)=\frac{e^{-(x-y)^2/(2t)}+e^{-(x+y)^2/(2t)}}{\sqrt{2\pi t}}.
\end{equation*}
Thus, in keeping with the remark immediately above the statement of the proposition, the $m$-particle wave function is the anti-symmetrised product of the one-particle solutions,
which is equivalent to \eqref{2.14}.
\end{proof}

The final task in proving~\eqref{2} is to determine the proportionality constant.

\begin{corollary}\label{cor2}
The matrix integral formula \eqref{2} holds true with
\begin{equation}\label{2.15}
c_m^\even =\prod_{k=1}^m\frac{\Gamma(2k-1)}{2}.
\end{equation}
\end{corollary}
\begin{proof}
The absolute value signs can be removed in~\eqref{2.13} due to the ordering of eigenvalues. Thus, combining \eqref{2.7} and \eqref{2.13} yields
\begin{equation}\label{cor2-proof-step1}
{1 \over (2\pi)^m }\operatorname{vol} \gO(2m) \int_{\gO(2m)}P_t^\even(RXR^T\vert \Lambda_Y)\,(R^TdR)
= \frac{\psi_t(x\vert y)}{\Delta_m^\even (x)\Delta_m^\even (y)}.
\end{equation}
Now set $t=1$. On the left-hand side, we know from~\eqref{2.1} that
\begin{equation*}
P_1^\even(RXR^T|\Lambda_Y)=\left(\frac{1}{2\pi}\right)^{m(2m-1)/2}\prod_{i=1}^m e^{-(x_i^2+y_i^2)/2}\,e^{\frac12\Tr RXR^T\Lambda_Y},
\end{equation*}
while on the right-hand side it follows from~\eqref{2.14} that
\begin{equation*}
\psi_1(x|y)=\left(\frac{1}{2\pi}\right)^{m/2}\prod_{i=1}^m e^{-(x_i^2+y_i^2)/2}\det[2\cosh (x_j y_k)]_{j,k=1}^m.
\end{equation*}
The corollary follows using these formulae in~\eqref{cor2-proof-step1} together with \cite{Mu82}
\begin{equation}\label{2.16}
\operatorname{vol} \gO(2m)=\prod_{k=1}^{2m}\frac{2\pi^{k/2}}{\Gamma(k/2)}
\end{equation}
and the duplication formula for gamma functions.
\end{proof}

\begin{remark}\label{rem1}
With $\int_{\gO(2m)}(R^TdR)=1$, taking the limit $X\to 0$ on the left-hand side of \eqref{2} gives unity. On the right-hand side, with $c_m^\even $ equal to \eqref{2.15}, we must also get unity by taking the limit $x_1,...,x_m\to 0$. The latter can be done by successive use of L'H\^opital's rule. The explicit form of the normalisation is also given in \cite[Eq.~(2)]{Zu16}.
\end{remark}

The derivation of \eqref{3} and \eqref{5} can be carried out along similar lines. Consider for definiteness \eqref{3}. The analogue of Proposition \ref{prop1} and Corollary \ref{cor1} is the following.

\begin{proposition}
Let $X$ and $Y$ be $(2m+1)\times (2m+1)$ real anti-symmetric matrices as in \eqref{3}. For $t>0$ define
\begin{equation}\label{2.17}
P_t^\odd (X|Y)=\left(\frac{1}{2\pi t}\right)^{m(2m+1)/2}e^{-\Tr(X-Y)^2/4t}
\end{equation}
and
\begin{multline}\label{2.18}
p_t^\odd (x|y)=\Delta_m^\even (x)\Delta_m^\odd (x)\operatorname{vol} \gO(2m+1)
\\
\times \left(\frac{1}{2\pi}\right)^m\int_{\gO(2m+1)} P_t^\odd (R\Lambda_X R^T|Y)\,(R^TdR),
\end{multline}
where $\{0\}\cup\{\pm ix_j\}_{j=1}^m$ denote the eigenvalues of $X$, ${\Delta}^\odd (x)$ is given in \eqref{4} and
\begin{equation}\label{2.19}
\Lambda_X=
\operatorname{diag}\left(0,
\Big[\begin{smallmatrix}0 & x_1 \\ -x_1 & 0\end{smallmatrix}\Big],
\ldots,
\Big[\begin{smallmatrix}0 & x_m \\ -x_m & 0\end{smallmatrix}\Big]
\right).
\end{equation}
The quantity $p_t^\odd $, which corresponds to eigenvalue PDF of $X$, satisfies the Fokker--Planck equation \eqref{2.4} with $p_t^\even $ replaced by $p_t^\odd $ and
\begin{equation}\label{2.20}
W=-\sum_{j=1}^m\log\lvert x_j\rvert-\sum_{1\leq j<k\leq m}\log\lvert x_j^2-x_k^2\rvert.
\end{equation}
Thus, we have
\begin{equation}\label{2.21}
\frac{1}{2}\left(\sum_{j=1}^m\frac{\partial^2}{\partial x_j^2}\right)\frac{p_t^\odd}{|\Delta_m^\odd (x)|}
=\frac{\partial}{\partial t}\frac{p_t^\odd}{|\Delta_m^\odd (x)|} ,
\end{equation}
with initial condition
\begin{equation*}
\left\lvert\frac{\Delta_m^\even (y)}{\Delta_m^\odd (x)}\right\rvert p_t^\odd (x|y)
\xrightarrow{t\to 0^+}\prod_{i=1}^m\delta(x_i-y_i)
\end{equation*}
and subject to the requirement that $p_t^\odd /\Delta_m^\odd (x)$ is anti-symmetric in $\{x_i\}_{i=1}^m$ and even in each $x_i$.
\end{proposition}

Solving \eqref{2.21} then gives us the analogue of Proposition \ref{prop2} and Corollary \ref{cor2}.

\begin{proposition}
With
\begin{equation}
\psi_t^\odd (x|y):=\left\lvert\frac{\Delta_m^\even (y)}{\Delta_m^\odd (x)}\right\rvert
p_t^\odd (x|y),
\end{equation}
we have
\begin{equation*}
\psi_t^\odd (x|y)
=\det\bigg[\frac{e^{-(x_j-y_k)^2/(2t)}{\color{red}{-}}e^{-(x_j+y_k)^2/(2t)}}{\sqrt{2\pi t}}\bigg]_{j,k=1}^m,
\end{equation*}
and consequently \eqref{3} holds true with
\begin{equation}\label{2.22}
c_m^\odd =\prod_{k=1}^m\frac{\Gamma(2k)}{2}.
\end{equation}
\end{proposition}

\section{Eigenvalues of  particular anti-symmetric  and   anti-self dual random matrices}\label{sec3}

\subsection{A particular Hermitian random matrix}

In \cite[Th.~6]{FR05}, the integration formula \eqref{1} was used to determine the eigenvalue PDF of the random matrix
\begin{equation}\label{3.1}
M'=UAU^\dagger +XBX^\dagger +\sqrt{t}Y.
\end{equation}
Here $U\in \gU(N)$ is chosen with invariant measure, $A$ is a fixed $N\times N$ complex Hermitian matrix, $B$ is a fixed $M\times M$ complex Hermitian matrix, $X$ is an $N\times M$ standard complex Gaussian matrix, $Y$ is a member of the $N\times N$ GUE defined with the diagonal elements
standard real Gaussians, while $t>0$ is a real parameter.
Specifically, with
\begin{equation}\label{3.2}
\mathcal{L}^{-1}[f(s)](x)=\lim_{\epsilon\to 0^+}\int_{i\R}e^{sx+\epsilon s^2/2}f(s)\,\frac{ds}{2\pi i},
\qquad x\in\R
\end{equation}
denoting the two-sided Laplace transform, it was shown that the eigenvalue PDF of \eqref{3.1}, with the eigenvalues denoted $\{\lambda_i\}_{i=1}^N$, is equal to
\begin{equation}\label{3.3}
\frac{1}{N!}\frac{\Delta_N(\lambda)}{\Delta_N(a)}
\det\bigg[\mathcal{L}^{-1}\Big[\frac{e^{ts^2/2}}{\det(\mathbb{I}+Bs)}\Big](\lambda_i-a_j)\bigg]_{i,j=1}^N.
\end{equation}
This, being of the form of a Vandermonde determinant $\Delta_N(\lambda)$ times the Slater--type determinant $\det[g_j(\lambda_i)]_{i,j=1}^N$, is an example of a so-called polynomial ensemble. The latter underlies many integrable structures in random matrix theory \cite{AKW13,AIK13,KS14,KZ14,AI15,FW15,FL15a,FIL17}. Below the matrix integrals \eqref{2}, \eqref{3} and \eqref{5} will be used to show that the natural anti-symmetric and anti-self dual analogues of \eqref{3.1} have for their eigenvalue PDF a polynomial ensemble form analogous to \eqref{3.3}.

\subsection{Even dimensional case: proof of Theorem \ref{prop4}}\label{sec:3.2}

The natural anti-symmetric analogue of \eqref{3.1} is (\ref{3.4}). Here we will make essential use of the Harish-Chandra matrix integral (\ref{2}) to give derivation of the claim of Theorem \ref{prop4} that its eigenvalue PDF is given by (\ref{3.5}).

Define the matrix-valued Fourier-Laplace transform of $M $ by
\begin{equation}\label{3.6a}
\hat f_M( C)=\mathbb{E}_M\big[e^{-\frac{1}{2}\Tr M  C}\big],
\end{equation}
where $C$ is a $2n\times 2n$ Hermitian anti-symmetric matrix. We note that in our case, the expectation with respect to $M$ is in fact an expectation respect to the triple $(\Omega,X,Y)$, which factorises. Thus
\begin{equation}\label{3.7}
\hat f_M (C)=
\mathbb{E}_\Omega\big[e^{-\frac{1}{2}\Tr \Omega A\Omega^T C}\big]
\mathbb{E}_X\big[e^{-\frac{1}{2}\Tr XBX^T C}\big]
\mathbb{E}_Y\big[e^{-\frac{\sqrt{t}}{2} \Tr YC}\big].
\end{equation}

Consider the first expectation in \eqref{3.7}. According to the orthogonal integral \eqref{2}, one has
\begin{equation}\label{3.8}
\mathbb{E}_\Omega\left(\exp(-\frac{1}{2}\Tr \Omega A\Omega^T C)\right)=\prod_{l=0}^{n-1}(2l)!\frac{\det[\cosh (a_ic_j)]_{i,j=1}^n}{{\Delta}^\even_n (a){\Delta}^\even_n (c)},
\end{equation}
where $\{c_j\}$ denotes the positive eigenvalues of the matrix $C$. In relation to the second expectation in~\eqref{3.7},
for convenience set $l = 2m$ (this does not effect the final expression). Then
$X=[x_{j,k}]_{j=1,...,2n}^{k=1,...,2m}$, and we write $\{b_k\}$ and $\{c_j\}$ for the positive eigenvalues of $B$ and $C$. Using the invariance of $X$ under orthogonal similarity transformation, we have
\begin{equation*}
\frac{1}{2}\Tr XBX^T C=\sum_{k=1}^m\sum_{j=1}^n b_kc_j(x_{2k-1,2j}x_{2k,2j-1}-x_{2k-1,2j-1}x_{2k,2j}).
\end{equation*}
Hence, by Gaussian integration
\begin{equation}\label{2nd-expect}
\mathbb{E}_X\left(\exp(-\frac{1}{2}\Tr XBX^T C)\right)
=\prod_{k=1}^m\prod_{j=1}^n(1-b_kc_j)^{-1}(1+b_kc_j)^{-1}
= \prod_{j=1}^n \det(\mathbb{I}+Bc_j)^{-1},
\end{equation}
where in last equality we recall that the eigenvalues of $B$ comes in pairs $\{\pm b_k\}$. Here, without loss of generality, we have assumed that $\{c_j\}$'s  satisfy $|b_k c_j|<1$ for all $j,k$; otherwise, we have to resort to the matrix-valued  Fourier transform as in  \eqref{3.6a}.
Finally, the third expectation in~\eqref{3.7} yields
\begin{equation}\label{3rd-expect}
\mathbb{E}_Y\left(\exp(-\frac{\sqrt{t}}{2} \Tr YC)\right)
=\mathbb{E}_Y\left(\exp(-\sqrt{t}\sum_{i=1}^n y_{2i-1,2i}c_i)\right)=\prod_{j=1}^n e^{\frac{1}{2}tc_j^2},
\end{equation}
with $Y=i[y_{j,k}]_{j,k=1,...,2n}$.

Substituting~\eqref{3.8}, \eqref{2nd-expect} and~\eqref{3rd-expect} into \eqref{3.7}, we thus have the evaluation
\begin{equation}\label{3.9}
\hat f_M(C)=\frac{\prod_{l=0}^{n-1}(2l)!}{{\Delta}_n^\even (a){\Delta}_n^\even (c)}
\det\bigg[\frac{e^{\frac{1}{2}tc_j^2}\cosh (a_i c_j)}{\det(\mathbb{I}+Bc_j)}\bigg]_{i,j=1}^n.
\end{equation}
On the other hand, noting the form \eqref{3.9} and Proposition \ref{propW1} in \S \ref{sect3.5} below, suppose $M$ is invariant under $2n\times 2n$ orthogonal change of basis, and has eigenvalue PDF
\begin{equation}\label{3.10}
\frac{1}{Z}{\Delta}_n^\even (\lambda)\det[g_i(\lambda_j)]_{i,j=1}^n, \qquad 0<\lambda_1<\cdots<\lambda_n.
\end{equation}
Then,
\begin{equation*}
\mathbb{E}_M\left(e^{-\frac12\Tr M C}\right)=\mathbb{E}_M\left(\int_{\gO(2n)} e^{-\frac12\Tr R M R^T C} (R^TdR)\right),
\end{equation*}
where $(R^TdR)$ is the normalised Haar measure on $\gO(2n)$. The expectation over $R$ is evaluated by the matrix integral \eqref{2}, and so by the ansatz~\eqref{3.10}
\begin{multline}
\mathbb{E}_M\left(e^{-\frac12\Tr M C}\right)
\\
= \frac{\prod_{l=0}^{n-1}(2l)!}{Z{\Delta}_n^\even (c)}
\idotsint\limits_{\ 0 < \lambda_1 < \cdots < \lambda_n} d\lambda_1 \cdots d\lambda_n\, \det[\cosh (\lambda_ic_j)]_{i,j=1}^n\det[g_j(\lambda_i)]_{i,j=1}^n.
\end{multline}
It follows by Andreief's integration formula (see e.g. \cite[Eq. (5.17c)]{Fo10}) that
\begin{equation}\label{3.11}
\mathbb{E}_M\left(e^{-\frac{1}{2}\Tr M C}\right)
=\frac{\prod_{l=0}^{n-1}(2l)!}{Z{\Delta}_n^\even (c)} \det\Big[\int_0^\infty g_k(x) \cosh (x c_j)\,dx\Big]_{j,k=1}^n.
\end{equation}
Comparison of \eqref{3.11} with \eqref{3.9} shows that the ansatz \eqref{3.10} is indeed correct, with the remaining task being to evaluate $g_k(\lambda)$. For this purpose, put $c_j\mapsto ic_j$ and suppose
\begin{equation*}
\int_0^\infty g_k(x) \cos (x c_j)\,dx=G_k(ic_j).
\end{equation*}
Taking the inverse cosine transform shows
\begin{equation}\label{3.12}
g_k(x)=\frac{2}{\pi}\int_0^\infty G_k(ic) \cos (x c)\,dc.
\end{equation}
But according to \eqref{3.9},
\begin{equation*}
G_k(ic_j)=\frac{e^{-\frac{1}{2}tc_j^2}\cos (a_k c_j)}{\det(\mathbb{I}+iBc_j)}.
\end{equation*}
Substituting in \eqref{3.12} and appropriately determining $Z$ as required by a comparison of \eqref{3.9} with \eqref{3.11} establishes \eqref{3.5}.

\subsection{Odd dimensional case}\label{sec:3.3}
In this case, let $\Omega$ be a $(2n+1)\times (2n+1)$ Haar distributed real orthogonal matrix, $A$ and $B$ be Hermitian anti-symmetric matrices of size $(2n+1)\times (2n+1)$ and $l\times l$ respectively. Also let $X$ be a $(2n+1) \times l$ real standard Gaussian matrix,  and let $Y$ be a $(2n+1)\times (2n+1)$ standard Gaussian Hermitian anti-symmetric matrix (PDF proportional to $e^{-\frac{1}{4}\Tr Y^2}$). We  define the $(2n+1)\times (2n+1)$  Hermitian anti-symmetric matrix
\begin{equation}\label{3.4odd}
M =\Omega A\Omega^T+XBX^T+\sqrt{t}Y,
\end{equation}
and proceed as in the even dimensional case. As in the latter, the following theorem does not depend on the parity of $l$; for convenience it will be assumed to be odd.

\begin{theorem}\label{prop4odd}
Let the positive eigenvalues of $A$ and $M$ be denoted $\{a_j\}_{j=1}^n$ and $\{\lambda_j\}_{j=1}^n$. Then the eigenvalue PDF of $M$ in \eqref{3.4odd} is equal to
\begin{equation}\label{3.5odd}
\frac{{\Delta}^\odd_n (\lambda)}{{\Delta}^\odd_n (a)}\det[g_k(\lambda_j)]_{j,k=1}^n, \qquad 0< \lambda_1<\lambda_2<\cdots <\lambda_n,
\end{equation}
where
\begin{equation}\label{3.6odd}
g_k(\lambda)=\frac{2}{\pi}\int_0^\infty e^{-\frac{1}{2}tc^2}\frac{\sin(a_k c)\,\sin(\lambda c)}{\det(\mathbb{I}+icB)}\,dc.
\end{equation}
\end{theorem}

\begin{proof}
Without loss of generality, assume that the eigenvalues of $A$ and $M$ are distinct and $l=2m+1$.
Define the matrix-valued Fourier-Laplace transform of $M $ by
\begin{equation}\label{3.6aodd}
\hat f_M( C)=\mathbb{E}_M\big[e^{-\frac{1}{2}\Tr M  C}\big],
\end{equation}
where $C$ is a  $(2n+1)\times (2n+1)$ Hermitian anti-symmetric matrix. By the independence of   $\Omega,X$ and $Y$, we get
\begin{equation}\label{3.7odd}
\hat f_M (C)=
\mathbb{E}_\Omega\big[e^{-\frac{1}{2}\Tr \Omega A\Omega^T C}\big]
\mathbb{E}_X\big[e^{-\frac{1}{2}\Tr XBX^T C}\big]
\mathbb{E}_Y\big[e^{-\frac{\sqrt{t}}{2} \Tr YC}\big].
\end{equation}

Consider the first expectation in \eqref{3.7odd}. According to the orthogonal integral \eqref{3}, one has
\begin{equation}\label{3.8odd}
\mathbb{E}_\Omega\left(\exp(-\frac{1}{2}\Tr \Omega A\Omega^T C)\right)=\prod_{l=0}^{n-1}(2l+1)!\frac{\det[\sinh(a_ic_j)]_{i,j=1}^n}{{\Delta}^\odd_n (a){\Delta}^\odd_n (c)},
\end{equation}
where $\{c_j\}$ denotes the positive eigenvalues of the matrix $C$. In relation to the second expectation in~\eqref{3.7}, using the orthogonal invariance of $X$, we may assume that
\begin{equation}
B=
\operatorname{diag}\left(
\Big[\begin{smallmatrix}0 & i b_1 \\ -ib_1 & 0\end{smallmatrix}\Big],
\ldots,
\Big[\begin{smallmatrix}0 & ib_m \\ -ib_m & 0\end{smallmatrix}\Big],0
\right), \quad C=
\operatorname{diag}\left(
\Big[\begin{smallmatrix}0 & i c_1 \\ -ic_1 & 0\end{smallmatrix}\Big],
\ldots,
\Big[\begin{smallmatrix}0 & ic_n \\ -ic_n & 0\end{smallmatrix}\Big],0
\right),
\end{equation}
for some positive $\{b_k\}_{k=1}^m$ and $\{c_j\}_{j=1}^{n}$.
It then follows that
\begin{equation*}
\frac{1}{2}\Tr XBX^T C=\sum_{k=1}^m\sum_{j=1}^n b_kc_j( x_{2j-1,2k-1}x_{2j,2k}-x_{2j-1,2k}x_{2j,2k-1}),
\end{equation*}
which implies
\begin{equation}\label{2nd-expectodd}
\mathbb{E}_X\left(\exp(-\frac{1}{2}\Tr XBX^T C)\right)
=\prod_{k=1}^m\prod_{j=1}^n(1-b_kc_j)^{-1}(1+b_kc_j)^{-1}
= \prod_{j=1}^n \det(\mathbb{I}+Bc_j)^{-1}.
\end{equation}
Finally, the third expectation in~\eqref{3.7odd} yields
\begin{equation}\label{3rd-expectodd}
\mathbb{E}_Y\left(\exp(-\frac{\sqrt{t}}{2} \Tr YC)\right)
=\mathbb{E}_Y\left(\exp(-\sqrt{t}\sum_{k=1}^n c_k y_{2k-1,2k})\right)=\prod_{k=1}^n e^{\frac{1}{2}tc_k^2},
\end{equation}
with $Y=i[y_{j,k}]_{j,k=1,...,2n+1}$.

Inserting~\eqref{3.8odd}, \eqref{2nd-expectodd} and~\eqref{3rd-expectodd} into \eqref{3.7odd}, we thus arrive at
\begin{equation}\label{3.9odd}
\hat f_M(C)=\frac{\prod_{l=0}^{n-1}(2l+1)!}{{\Delta}_n^\odd (a){\Delta}_n^\odd (c)}
\det\bigg[\frac{e^{\frac{1}{2}tc_j^2}\sinh (a_k c_j)}{\det(\mathbb{I}+Bc_j)}\bigg]_{j,k=1}^{n}.
\end{equation}
On the other hand, noting the equation \eqref{3.9odd} and  \S \ref{sect3.5} below, we may suppose that $M$ is invariant under $(2n+1)\times (2n+1)$ orthogonal change of basis, and has eigenvalue PDF
\begin{equation}\label{3.10odd}
\frac{1}{Z}{\Delta}_n^\odd (\lambda)\det[g_k(\lambda_j)]_{k,j=1}^n, \qquad 0<\lambda_1<\cdots<\lambda_n.
\end{equation}
Then, by the matrix integral \eqref{3}, we can evaluate
\begin{equation*}
\mathbb{E}_M\left(e^{-\frac12\Tr M C}\right)=\mathbb{E}_M\left(\int_{\gO(2n+1)} e^{-\frac12\Tr R M R^T C} (R^TdR)\right),
\end{equation*}
and get under the ansatz~\eqref{3.10odd}
\begin{equation}\label{3.11odd}
\mathbb{E}_M\left(e^{-\frac12\Tr M C}\right)
=\frac{\prod_{l=0}^{n-1}(2l+1)!}{Z{\Delta}_n^\odd (c)} \det\Big[\int_0^\infty g_k(x) \sinh(x c_j)\,dx\Big]_{j,k=1}^n.
\end{equation}
Comparing \eqref{3.11odd} with \eqref{3.9odd}, we choose $Z = \Delta_n^{\rm odd}(a)$ so that the factors outside the determinant match. To determine the function $g_k$ so that the entries of the determinant match, we put $c_j\mapsto ic_j$ to obtain
\begin{equation*}
\int_0^\infty g_k(x) \sin (x c_j)\,dx=\frac{e^{-\frac{1}{2}tc_j^2}\sin(a_k c_j)}{\det(\mathbb{I}+ic_j B)}.
\end{equation*}
Taking the inverse sine transform shows
\begin{equation}\label{3.12odd}
g_k(x)=\frac{2}{\pi}\int_0^\infty   \sin(x c)\, \frac{e^{-\frac{1}{2}tc^2}\sin(a_k c)}{\det(\mathbb{I}+ic B)}dc,
\end{equation}
and the statement of the theorem follows.
\end{proof}

\subsection{Anti-self dual case}\label{sec:3.4}
We recall that an even dimensional matrix $X$ is said to be symplectic if
$$
X^T J_{2N} X = J_{2N}, \qquad J_{2N} = \begin{bmatrix} 0_{N } & \mathbb I_{N } \\
-  \mathbb I_{N }  &
0_{N } \end{bmatrix}.
$$
Let $\Omega$ be a $2n \times 2n$ unitary symplectic matrix. It must then have the block structure
\begin{equation}\label{XZW}
\Omega = \begin{bmatrix} Z & W \\ - \overline{W} & \overline{Z} \end{bmatrix},
\end{equation}
where $W$ and $Z$ are of size $n \times n$. With $X_1, X_2$     $n \times m$ complex standard Gaussian matrices, let
$X$ be a  $2\times 2$ block matrix
\begin{equation}\label{def:X}
X=\begin{bmatrix}
X_1 & X_2 \\
-\overline{X_2} & \overline{X_1}
\end{bmatrix};
\end{equation}
cf.~(\ref{XZW}).

For $A_j$ ($j=1,\dots,4$) $n \times n$ complex matrix, define the dual operation {\small $\rm D$} by
$$
\begin{bmatrix} A_1 & A_2 \\
A_3 & A_4
\end{bmatrix}^{\rm D}
=
\begin{bmatrix} A_4^T & -A_2^T \\
-A_3^T & A_1^T
\end{bmatrix}.
$$
A block matrix $\tilde{H}$ is said to be anti-self dual if $$\tilde{H} = - \tilde{H}^{\rm D},$$
and thus
$$
\tilde{H} = \begin{bmatrix} H_1 & H_2 \\
- \overline{H_2} & \overline{H_1} \end{bmatrix}, \quad H_{1}^{\dagger}=-H_{1}, \quad H_{2}^{T}=H_{2}.
$$
The matrix $$H := i \tilde{H}$$ is then termed Hermitian anti-self dual.
Define $A$ and $B$  to be Hermitian anti-self dual   matrices of size $ 2n\times 2n$ and $2m\times 2m$ respectively. Also, let $Y$ be a $2n \times 2n$ Hermitian anti-self dual matrix with PDF
proportional to $\exp(-\frac{1}{4} {\rm Tr} \, Y^2)$.

In terms of the above defined matrices, define the $2n\times 2n$ Hermitian anti-self dual matrix
\begin{equation}\label{3.4selfdual}
M =\Omega A\Omega^\dagger+XBX^\dagger+\sqrt{t}Y.
\end{equation}
As for the analogous construction in the case of Hermitian anti-symmetric matrices, our
interest is in the eigenvalue PDF of $M$.

\begin{theorem}\label{prop4selfdual}
Let the positive eigenvalues of $A$ and $M$ be denoted $\{a_j\}_{j=1}^n$ and $\{\lambda_j\}_{j=1}^n$. Then the eigenvalue PDF of $M$ in \eqref{3.4selfdual} is equal to
\begin{equation}\label{3.5dual}
\frac{{\Delta}^\odd_n (\lambda)}{{\Delta}^\odd_n (a)}\det[g_k(\lambda_j)]_{j,k=1}^n, \qquad 0 < \lambda_1<\lambda_2<\cdots <\lambda_n,
\end{equation}
where
\begin{equation}\label{3.6dual}
g_k(\lambda)=\frac{2}{\pi}\int_0^\infty e^{-\frac{1}{2}tc^2}\frac{\sin (a_k c)\,\sin(\lambda c)}{\det(\mathbb{I}+ icB)}\,dc.
\end{equation}
\end{theorem}

\begin{proof}
For $C$ being a  $2n\times 2n$ Hermitian anti-self dual matrix, let
\begin{equation}\label{3.6dual}
\hat f_M( C)=\mathbb{E}_M\big[e^{-\frac{1}{2}\Tr M  C}\big],
\end{equation}
then the independence of $\Omega,X$ and $Y$ implies that
\begin{equation}\label{3.7dual}
\hat f_M (C)=
\mathbb{E}_\Omega\big[e^{-\frac{1}{2}\Tr \Omega A\Omega^\dagger C}\big]
\mathbb{E}_X\big[e^{-\frac{1}{2}\Tr XBX^\dagger C}\big]
\mathbb{E}_Y\big[e^{-\frac{\sqrt{t}}{2} \Tr YC}\big].
\end{equation}

Consider the first expectation in \eqref{3.7dual}, according to the   integral \eqref{5} one has
\begin{equation}\label{3.8dual}
\mathbb{E}_\Omega\left(\exp(-\frac{1}{2}\Tr \Omega A\Omega^\dagger C)\right)=\prod_{l=0}^{n-1}(2l+1)!\frac{\det[\sinh(a_ic_j)]_{i,j=1}^n}{{\Delta}^\odd_n (a){\Delta}^\odd_n (c)},
\end{equation}
where $\{c_j\}$ denotes the positive eigenvalues of the matrix $C$.
For the second expectation, using the invariance of $X$,
we can assume that
\begin{equation}
B=
\operatorname{diag}\left(
  b_1, \ldots,  b_m,  - b_1, \ldots,  -b_m\right), \qquad C= \operatorname{diag}\left(
  c_1, \ldots,  c_n,  -c_1, \ldots,  -c_n\right),
\end{equation}
for some positive $\{b_k\}_{k=1}^m$ and $\{c_j\}_{j=1}^{n}$.
Recall the definition of $X$ given in \eqref{def:X}, with $X_{l}=[x^{(l)}_{jk}], l=1,2$,  we have
\begin{equation*}
\frac{1}{2}\Tr XBX^\dagger C=\sum_{k=1}^m\sum_{j=1}^n b_kc_j\big( |x^{(2)}_{jk}|^2-|x^{(1)}_{jk}|^2\big),
\end{equation*}
which implies
\begin{equation}\label{2nd-expectdual}
\mathbb{E}_X\left(\exp(-\frac{1}{2}\Tr XBX^\dagger C)\right)
= \prod_{j=1}^n \det(\mathbb{I}+Bc_j)^{-1}.
\end{equation}
Finally, the third expectation in~\eqref{3.7dual} yields
\begin{equation}\label{3rd-expectdual}
\mathbb{E}_Y\left(\exp(-\frac{\sqrt{t}}{2} \Tr YC)\right)
=\mathbb{E}_Y\left(\exp(\sqrt{t}\sum_{k=1}^n c_k y^{(1)}_{k,k})\right)=\prod_{k=1}^n e^{\frac{1}{2}tc_k^2},
\end{equation}
with $Y=i\begin{bmatrix} Y_1 & Y_2 \\
- \overline{Y_2} & \overline{Y_1} \end{bmatrix}$ and $Y_1=[y^{(1)}_{j,k}]_{j,k=1,...,n}$.

Inserting~\eqref{3.8dual}, \eqref{2nd-expectdual} and~\eqref{3rd-expectdual} into \eqref{3.7dual} then gives us
\begin{equation}\label{3.9dual}
\hat f_M(C)=\frac{\prod_{l=0}^{n-1}(2l+1)!}{{\Delta}_n^\odd (a){\Delta}_n^\odd (c)}
\det\bigg[\frac{e^{\frac{1}{2}tc_j^2}\sinh (a_k c_j)}{\det(\mathbb{I}+c_jB)}\bigg]_{j,k=1}^{n}.
\end{equation}
As in the Hermitian anti-symmetric  matrices case,  \eqref{3.5dual} follows by taking the inverse transform and this completes the proof of the theorem.
\end{proof}

As for the evaluation of the respective Harish-Chandra group integrals (\ref{3}) and (\ref{5}), we
see that the eigenvalue PDFs of (\ref{3.4selfdual}) and (\ref{3.4odd}) are identical.

\subsection{Transform structure} \label{sect3.5}

Kieburg and K\"osters have recently introduced the matrix spherical transform to the analysis of
both the eigenvalue and singular value distribution of products of unitary invariant matrix ensembles
\cite{KK16a,KK16b,FKK17,Ki17}. It was pointed out by Kuijlaars and Roman \cite{KR16} that the analogue in the
case of sums of unitary invariant matrix ensembles is the matrix Fourier transform. The latter
is equivalent to the matrix-valued  Fourier--Laplace transform (\ref{3.6a}).
Given these developments, it is then of some interest to put the workings of
\S~\ref{sec:3.2} (similar ideas apply to \S~\ref{sec:3.3} and \ref{sec:3.4}) into an analogous framework.

\begin{lemma}
Let $X$ and $C$ be $2n \times 2n$ Hermitian anti-symmetric matrices, we write $X = i [x_{jk} ]_{j,k=1}^{2n}$ and
$C = i [c_{jk} ]_{j,k=1}^{2n}$. Furthermore, let $X$ be random with distribution having PDF
$f(X)$, and define $\hat{f}_X(C)$ by the matrix Fourier--Laplace transform~\eqref{3.6a}.
With  $(dC)=\prod_{1\leq j<k\leq 2n}dc_{jk}$, we have
\begin{equation}\label{w.0}
f(X) = \Big ( \frac{1}{2\pi} \Big )^{n(2n-1)} \int \hat{f}_X(iC) \exp \big ( \frac{i}{2}
{\rm Tr} \, X C \big ) \, (d C).
\end{equation}
\end{lemma}

\begin{proof}
We have by definition
\[
\hat{f}_X(C) = \mathbb E_X\big[ e^{-\frac12\Tr XC}\big] =
 \int_{\R^{n(2n-1)}} \prod_{1\leq j<k\leq2n} \,e^{-x_{jk}c_{jk}}f(X)\,(dX).
\]
Multiplying both sides  by
\[
 \exp \Big ( \frac{1}{2} {\rm Tr} \, XC \Big ) = \exp \Big ( \sum_{j<k} x_{jk} c_{jk} \Big ),
\]
then integrating both sides over the imaginary axis in the complex $c_{jk}$-hyperplane $(1 \le j < k \le
2n)$ gives, by the usual multivariate inverse Fourier transform formula, (\ref{w.0}).
\end{proof}

Of particular interest is the case that the transform $\hat{f}_X(C)$ is unchanged by conjugation with real
orthogonal matrices.

\begin{proposition}\label{propW1}
Let $C$ and $X$ be $2n\times 2n$ Hermitian anti-symmetric matrix with positive eigenvalues $c = \{ c_j \}_{j=1}^n$ and $x = \{ x_j \}_{j=1}^n$, respectively. Assume that
$$\hat{f}_X(iC) = \hat{f}_X(i R C R^T)$$
for all $R \in\gO(2n)$. Write
\[
 \hat{f}_X(c) := \hat{f}_X \bigg (  {\rm diag} \, \begin{bmatrix} 0 & i c_j \\ - i c_j & 0 \end{bmatrix}_{j=1}^n \bigg ),
 \]
with  $\hat{f}_X(ic)$  being assumed to decay exponentially
fast in each $c_j$,   and denote by $f(x)$ the positive eigenvalue PDF of $X$.
We have
 \begin{equation}\label{w.2}
 f(x) = \prod_{j=0}^{n-1}\frac{2}{\pi(2j)!}
 \Delta^{\even}_n(x) \int_{0}^\infty dc_1 \cdots  \int_{0}^\infty dc_n \,
 \hat{f}_X(ic)
 \Delta^{\even}_n(ic) \prod_{j=1}^n \cos(x_j c_j).
 \end{equation}
 \end{proposition}

 \begin{proof}
 Using the orthogonal invariance of $\hat{f}(iC)$ and the measure $(dC)$, it follows from~\eqref{w.0} that
 \[
 f(X) = \Big ( \frac{1}{2 \pi} \Big )^{n(2n-1)} \int \hat{f}_X(iC)
 \int_{\gO(2n)}  e^{\frac i2 {\Tr} R X R^T C} (R^TdR)\, (d C),
 \]
 where $(R^TdR)$ is the normalised Haar measure on $\gO(2n)$. Here the integration
 over the matrices $C$ has been commuted to occur after
 that over the matrices $R$; this is justified by the assumed
 fast decay of $\hat{f}_X(ic)$. 
 Now, using the matrix integral (\ref{2}) yields
 \[
 f(X) = \frac{1}{(2 \pi)^{n (2n-1)}}
 \frac{\prod_{l=0}^{n-1} (2 l)! }{ \Delta^{\rm even}_n(x)}
 \int \hat{f}_X(ic) \frac{\det [ \cos (x_j c_k) ]_{j,k=1}^n }{ \Delta^{\rm even}_n(ic)} \, (d C).
 \]
The integrand depends only on the eigenvalues and eigenvectors, with the latter decoupling
and only contributing a constant. In this setting, one has for the Jacobian
(see e.g.~\cite[Eq.~(2.18)]{FN09})
\begin{equation}\label{w.1a}
\frac{1}{ Z_n} \Delta^{\rm even}_n(c)^2, \qquad
Z_n =  \pi^{n} (2\pi)^{-n^2} \prod_{j=0}^{n-1} (2j)!.
\end{equation}
Hence
\begin{align*}
{f}(X) & =  \frac{2^n}{(2 \pi)^{n^2}} \frac{1}{\Delta^{\rm even}_n(x)}
\int_{0< c_1< \cdots < c_n} \hat{f}_X(ic) \Delta^{\rm even}_n(ic) \det [ \cos (x_j c_k) ]_{j,k=1}^n dc_1\cdots dc_n\\
& = \frac{2^n}{(2 \pi)^{n^2}} \frac{1}{ \Delta^{\rm even}_n(x)}
\int_0^\infty \cdots  \int_0^\infty \,  \hat{f}_X(ic) \Delta^{\rm even}_n(ic)
\prod_{j=1}^n \cos (x_j c_j)\, dc_1\cdots dc_n.
\end{align*}
Here, the second line follows by first noting that since the integrand is symmetric, the integration
range can be taken to be $\mathbb R^n_+$, provided we divide by $n!$, then noting that since both
$\Delta^{\rm even}_n(c)$ and the determinant are anti-symmetric, only the latter can be
replaced by $n!$ times the diagonal term.

Now, to deduce the eigenvalue PDF of $X$, it only remains to multiply the right-hand side by the Jacobian
(\ref{w.1a}), with $c$ therein replaced by $x$, and (\ref{w.2}) results.
\end{proof}

An immediate (but significant) property of the matrix Fourier--Laplace transform (\ref{3.6a})
is its factorisation with respect to the addition of matrices.

\begin{lemma}\label{lemmaW2}
Let $X$ and $Y$ be independent $2n \times 2n$ Hermitian anti-symmetric random matrices, then we have
\begin{equation}
\hat{f}_{X+Y}(C) = \hat{f}_{X }(C) \hat{f}_{Y}(C).
\end{equation}
\end{lemma}

\begin{proof}
Identically to the scalar case, we have
\[
\hat{f}_{X+Y}(C)  :=
 \mathbb E_{X,Y}\big[e^{\frac12\Tr (X+Y) C}\big]=\mathbb E_{X,Y}\big[e^{\frac12\Tr X C}e^{\frac12\Tr Y C}\big]=
 \mathbb E_{X}\big[e^{\frac12\Tr X C}]\mathbb E_{Y}\big[e^{\frac12\Tr Y C}\big],
\]
where the last equality follows by independence.
\end{proof}

We are now in a position to give an alternative derivation of the result of Theorem \ref{prop4}.
With $M$ given by (\ref{3.4}), in accordance with Lemma \ref{lemmaW2},
$\hat{f}_M$ has been computed above to be equal to (\ref{3.9}).
Substituting in
(\ref{w.2}) we see that the integration can be done row-by-row to reclaim (\ref{3.5}).

\section{Special cases for even dimensions}\label{sec4}


\subsection{A low rank case}
In this subsection we will consider the matrix~\eqref{3.4} with $t=0$ and $l=2$.
Inspection of \eqref{3.4} shows that for these parameters the eigenvalue problem for $M $ is identical to the eigenvalue problem for
\begin{equation}\label{3.13}
\operatorname{diag}(a_1,\ldots,a_n)\otimes
\begin{bmatrix}
0 & i \\
-i & 0
\end{bmatrix}
+ X_{2n\times 2}
\begin{bmatrix}
0 & ib_1 \\
-ib_1 & 0
\end{bmatrix}(X^T)_{2\times 2n}.
\end{equation}
This corresponds to a rank two perturbation of the anti-symmetric matrix
\[
\operatorname{diag}(a_1,\ldots,a_n)\otimes
\begin{bmatrix} 0 & i \\ -i & 0 \end{bmatrix}.
\]
We note that rank two is the lowest possible rank for a perturbation which preserves the matrix structure.

\begin{corollary}\label{cor4.1}
In the case $t=0$, $l=2$, the random matrix \eqref{3.4} has the same eigenvalue PDF as \eqref{3.13}. This eigenvalue PDF is given by \eqref{3.5} with
\begin{equation}\label{3.14}
g_i(\lambda)=\frac{1}{2b_1}\left(e^{-(\lambda+a_i)/b_1}+e^{-\lvert\lambda-a_i\rvert/b_1}\right).
\end{equation}
\end{corollary}
\begin{proof}
According to \eqref{3.6}, in the case $t=0$, $l=2$,
\begin{equation*}
g_i(\lambda)=\frac{2}{\pi}\int_0^\infty \frac{\cos (a_ic) \cos (\lambda c)}{1+b_1^2c^2}\,dc.
\end{equation*}
Evaluating this integral gives \eqref{3.14}.
\end{proof}

The working of \S \ref{sec2} tells us that we require a definite ordering of $\{\lambda_i\}$, say $0<\lambda_1< \cdots <\lambda_n$, for \eqref{3.5} to be correctly normalised. With $\{a_i\}$ similarly ordered, denote by $a \prec \lambda$ the interlacing
\begin{align*}
a_1<\lambda_1<a_2<\lambda_2<\cdots <a_n<\lambda_n.
\end{align*}
A feature of $\det[g_i(\lambda_j)]_{i,j=1}^n$ with $g_i(\lambda)$ given by \eqref{3.14} is that it vanishes unless $a \prec \lambda$ or $\lambda \prec a$, when in the former case it reduces to (after setting $b_1=1$ for convenience)
\begin{align*}
2^{-n}(e^{-(\lambda_1+a_1)}+e^{-\lvert \lambda_1-a_1\rvert})e^{-\sum_{j=1}^n(a_j+\lambda_j)}\prod_{l=2}^n(e^{2a_l}-e^{\lambda_{l-1}}),
\end{align*}
and in the latter case to the same expression with each $a_l$ interchanged with $\lambda_l$ for $\lambda \prec a$. In fact a result of Defosseux \cite{De10,De13} tells us that for all orderings, and with $g_i(\lambda)$ given by \eqref{3.14},
\begin{multline}\label{X0}
\det[g_i(\lambda_j)]_{i,j=1}^n= \\
\frac{1}{2}\int
\chi({z\prec\lambda})\chi({z\prec a}) e^{-\sum_{j=1}^n(a_j+\lambda_j - 2 z_j)}(e^{-(\lambda_1+a_1)}+e^{-\lvert \lambda_1-a_1\rvert})\,dz_1 \cdots dz_{n-1},
\end{multline}
where $\chi(T)=1$ for $T$ true and $\chi(T)=0$ otherwise. Here, $z\prec \lambda$ is to be interpreted as though there is a variable $z_0:=0$ and thus $0<\lambda_1<z_1< \cdots <z_{n-1}<\lambda_n$; similarly for $z\prec a$.

Let $A$ and $B$ be $N\times N$ Hermitian matrices with $B$ a rank-$r$ perturbation of $A$. With $\{\sigma_i(X)\}_{i=1}^N$ denoting the singular values of the matrix $X$, ordered so that $0\leq \sigma_1(X)\leq \cdots \leq \sigma_N(X)$, it is a known theorem that \cite{Th76}
\begin{equation}\label{X1}
\sigma_{k-r}(A+B)\leq \sigma_k(A)\leq \sigma_{k+r}(A+B).
\end{equation}
Since for $2n\times 2n$ anti-symmetric matrices the singular values correspond to the $n$ eigenvalues $0\leq \lambda_1\leq \cdots\leq \lambda_n$, we see that the conditions for \eqref{X0} to be nonzero are consistent with \eqref{X1} in the case $r=2$.

\subsection{Hermitised products}
In this subsection, we consider the matrix~\eqref{3.4} with $t=0$ and $A=0$. While we can immediately set $t=0$ in \eqref{3.6}, the structure of \eqref{3.5} does not permit us to immediately set each $a_i=0$. Instead a sequential limiting procedure, analogous to that in Remark \ref{rem1}, is required. Performing these steps shows that the PDF in Theorem \ref{prop4} reduces to
\begin{equation}\label{3.15}
\prod_{l=1}^{n-1}\frac{1}{(2l)!}{\Delta}_n^\even (\lambda)
\det\bigg[\frac{2}{\pi}\int_0^\infty \frac{(-c^2)^{i-1} \cos (\lambda_j c)}{\det(\mathbb{I}+icB)}\,dc\bigg]_{i,j=1}^n.
\end{equation}
The task now is to simplify the determinant. For all integrals therein to be well-defined, we assume that $l=2m$ and $m\geq n$, since $\det(\mathbb{I}+icB)=\prod_{l=1}^m (1+c^2b_l^2)$.

\begin{proposition}\label{prop5}
In the case $l=2n$, the eigenvalue PDF \eqref{3.15} corresponding to $t=0$, $A=0$ in Theorem \ref{prop4}
and thus of the random matrix $XBX^T$ reduces to
\begin{equation}\label{3.16}
\prod_{l=1}^{n-1}\frac{1}{(2l)!} \frac{{\Delta}_n^\even (\lambda)}{{\Delta}_n^\odd (b)}\det[e^{-\lambda_j/b_k}]_{j,k=1}^n, \qquad \lambda_1<\lambda_2<\cdots<\lambda_n.
\end{equation}
\end{proposition}
\begin{proof}
For the integral in \eqref{3.15} we have
\begin{equation}
\frac{2}{\pi}\int_0^\infty \frac{(-c^2)^{i-1}\cos (\lambda_j c)}{\det(\mathbb{I}+icB)}\,dc
=\frac{1}{\pi}\int_{-\infty}^\infty \frac{(-c^2)^{i-1}e^{i\lambda_j c}}{\prod_{l=1}^m (1+(cb_l)^2)}\,dc
 = \sum_{k=1}^m \frac{e^{-\lambda_j/b_k}b_k^{2(m-i-1/2)}}{\prod_{l=1,l\neq k}^m(b_k^2-b_l^2)},
\end{equation}
where the final equality follows by closing the contour in the upper half plane and evaluating using residues.

Substituting in \eqref{3.15} shows that in the special case $l=2n$ the determinant therein can be written
\begin{multline*}
\frac{1}{{\Delta}_n^\even (b){\Delta}_n^\odd (b)}\sum_{k_1,...,k_n=1}^n \prod_{j=1}^n e^{-\lambda_j/b_{k_j}}\det[b_{k_j}^{2(i-1)}]_{i,j=1}^n
\\
= \frac{\det[e^{-\lambda_j/b_k}]_{j,k=1}^n \det[b_{k_j}^{2(i-1)}]_{i,j=1}^n}{{\Delta}_n^\even (b){\Delta}_n^\odd (b)},
\end{multline*}
where the equality follows by noting that because the determinant in the first line is anti-symmetric in $\{b_k\}$, the sum can be replaced by the anti-symmetrisation of $\prod_{j=1}^n e^{-\lambda_j/b_j}$ times $\det[b_j^{2(i-1)}]_{i,j=1}^n$. The proposition follows by noting that $\det[b_j^{2(i-1)}]_{i,j=1}^n$ cancels with ${\Delta}_n^\even (b)$.
\end{proof}

One point of interest is the limiting case of Proposition \ref{prop5} in which some of $b_1,\dots,b_n$ tend to zero.

\begin{corollary}[Defosseux \cite{De10}]\label{cor4}
Let $X$ be a $2n\times 2m$ ($n\geq m$) standard real Gaussian matrix. With $b_j>0$ ($j=1,\dots,m$), let
\begin{align*}
B=\operatorname{diag}(b_1,\ldots,b_m)\otimes
\begin{bmatrix}
0 & i \\
-i & 0
\end{bmatrix}.
\end{align*}
The eigenvalue PDF of $XBX^T$, or equivalently of $\tilde{X} \tilde{B} \tilde{X}^T$, where
$\tilde{X}$ is a $2n \times 2n$ standard Gaussian matrix and $\tilde{B}$ is the $n \times n$
block diagonal matrix
\[
\tilde{B} =
\operatorname{diag}\Big(
\big[\begin{smallmatrix} 0 & ib_1 \\ -ib_1 & 0 \end{smallmatrix}\big],
\ldots,
\big[\begin{smallmatrix} 0 & ib_m \\ -ib_m & 0 \end{smallmatrix}\big],
\underbrace{%
\big[\begin{smallmatrix} 0 & 0 \\ 0 & 0\vphantom{b_m} \end{smallmatrix}\big]
\ldots,
\big[\begin{smallmatrix} 0 & 0 \\ 0 & 0\vphantom{b_m} \end{smallmatrix}\big]%
}_{n-m\textup{ times}}
\Big),
\]
 is equal to
\begin{equation}\label{3.18}
\prod_{l=n-m}^{n-1}\frac{1}{(2l)!}\frac{{\Delta}_m^\even (\lambda)}{{\Delta}_m^\odd (b)} \prod_{l=1}^{m}(\lambda_l/b_l)^{2(n-m)} \det[e^{-\lambda_j/b_k}]_{j,k=1}^{m}, \qquad \lambda_1<\lambda_2<\cdots<\lambda_n.
\end{equation}
\end{corollary}
\begin{proof}
With the $\lambda_i$ ordered so that $\lambda_1$ is the smallest etc., we see that in the limit $b_1\to 0^+$ the determinant in \eqref{3.16} has the leading order value requiring $\lambda_1\to 0^+$ and equalling the top left entry $e^{-\lambda_1/b_1}$ times  the determinant of the $(n-1)\times (n-1)$ sub-block obtained by deleting the first row and first column. Noting too that
\begin{align*}
\frac{\Delta_n^\even (\lambda)}{\Delta_n^\odd (b_j)}
\to \frac{1}{b_1}\prod_{l=2}^{n}(\lambda_l/b_l)^2
\frac{\Delta_{n-1}^\even (\{\lambda_j\}_{j=2}^{n})}{\Delta_{n-1}^\odd (\{b_j\}_{j=2}^{n})},
\end{align*}
and that $\frac{1}{b_1}e^{-\lambda_1/b_1}$ tends to the delta function $\delta(\lambda_1)$, we obtain \eqref{3.18} in the case $m=n-1$ after renaming $\{\lambda_j\}_{j=2}^n \mapsto \{\lambda_j\}_{j=1}^{n-1}$ and similarly for $\{b_j\}$.
Repeating this procedure starting with the case $m=n-1$ gives the result for $m=n-2$, etc.
\end{proof}

The structure \eqref{3.18} further simplifies by setting $b_1= \dots =b_{m}=1$, which we do sequentially in an analogous manner to that specified in Remark \ref{rem1}.

\begin{corollary}
Let $X$ be an $2n\times 2m$ ($n\geq m$) standard real Gaussian matrix. The eigenvalue PDF of
\begin{align*}
X \left(\mathbb I_m\otimes
\begin{bmatrix} 0 & i \\ -i & 0 \end{bmatrix}\right)
X^T
\end{align*}
is equal to
\begin{equation}\label{3.19}
\frac{1}{2^{m(m-1)/2}}\prod_{l=n-m}^{n-1}\frac1{(2l)!}\prod_{l=0}^{m-1}\frac1{l!}
\prod_{l=1}^{m}\lambda_l^{2(n-m)}e^{-\lambda_l}\prod_{1\leq j<k\leq m}(\lambda_k-\lambda_j)(\lambda_k^2-\lambda_j^2).
\end{equation}
\end{corollary}

\begin{remark}
In the case $m=n$, this result was first derived in \cite{LSZ06}. The functional form is a particular example of the so-called Laguerre Muttalib-Borodin ensemble \cite{Mu95,Bor99,FW15}.
\end{remark}

Eigenvalue PDFs of the form (\ref{3.18}) for the random matrix structure $X B X^T$, with $B$ fixed,
underlie the fact  that the singular value PDF for certain product ensembles are determinantal
point processes of the polynomial ensemble type; see the text below (\ref{3.3}).
The essential mechanism is that in the case the eigenvalue PDF of $B =: B_0$ is of the
polynomial ensemble form
\begin{equation}\label{f1}
\frac{1}{Z_m} \Delta_m^\even (b_j )
\det [ g_{j-1}(b_k) ]_{j,k=1}^m
\end{equation}
for some $\{g_j(x)\}$, (\ref{3.18}) tells us that $B_1 := X B_0 X^T$ also gives rise to a polynomial ensemble.
This follows immediately from Andreief's integration formula. Iterating allows
for the determination of the PDF of the eigenvalues of
\begin{equation}\label{f1a}
B_j = X_j B_{j-1} X_j^T,
\end{equation}
where $X_j$ is  of size $2(m + \nu_j) \times 2(m+ \nu_{j-1})$ with $\nu_j \ge \nu_{j-1}$
$(j=1,2,\dots)$ and $\nu_0 = 0$.

\begin{corollary}
In the setting of Corollary \ref{cor4}, suppose $\{b_j\}_{j=1}^m$, assumed ordered, have PDF (\ref{f1}).
The eigenvalues of the random matrix $X B X^T$, or equivalently of $\tilde{X} \tilde{B}
\tilde{X}^T$, have PDF
\begin{equation}\label{f2}
\frac{1}{Z_m} \prod_{l=n-m}^{n-1} \frac{1}{(2 l)!}
\Delta^\even_m (\lambda)
\det \Big [ \int_0^\infty e^{-b} b^{2(n-m) - 1} g_{j-1} \Big ( \frac{\lambda_k}{b} \Big ) \, d b
\Big ]_{j,k=1}^m.
\end{equation}
Moreover, a general member
\begin{equation}\label{f3}
B_M = X_M \cdots X_1 B_0 X_1^T \cdots X_M^T
\end{equation}
of the sequence of random matrices $B_1,B_2,\dots$ as specified by (\ref{f1a}) has PDF for its
eigenvalues given by
\begin{equation}\label{f4}
\frac{1}{Z_m} \prod_{j=1}^M\prod_{l=\nu_j}^{m+\nu_j-1}  \frac{1}{(2l)!}\Delta^\even_m (\lambda)
\det \Big [ g_{j-1}^{(M)}(\lambda_k) \Big ]_{j,k=1}^m,
\end{equation}
where $g_j^{(M)}$ is defined recursively according to
\begin{equation}\label{f5}
g_j^{(s)}(\lambda) = \int_0^\infty e^{-b} b^{2\nu_s - 1} g_j^{(s-1)} \Big ( \frac{\lambda}{b} \Big ) \, {\rm d}b
\end{equation}
with $g_j^{(0)}(\lambda) := g_j(\lambda)$.
\end{corollary}

\subsection{Initial matrix as an elementary anti-symmetric matrix and proof of Corollary \ref{cor1even}}
\label{sec:ele}

In this subsection, we consider the case where the initial matrix in~\eqref{f3} is given by
\[
B_0 = \mathbb I_m\otimes  \begin{bmatrix} 0 & i \\ -i & 0 \end{bmatrix}.
\]
In this case, according to (\ref{3.19}), the random matrix $X_1 B_0 X_1^T$ has for its eigenvalues
the PDF
\[
\frac{1}{Z_m}\prod_{l=\nu_1}^{m+\nu_1-1} \frac{1}{(2l)!} \Delta^\even_m (\lambda)
\det \Big [ \lambda_k^{2 \nu_1 +j-1} e^{-\lambda_k}   \Big ]_{j,k=1}^m,
\]
with $$Z_m = 2^{m(m-1)/2} \prod_{l=0}^{m-1} l!,$$
and is thus of the form (\ref{f4}) with $M=1$ and
\begin{equation}\label{f6}
g_{j-1}^{(1)}(\lambda) = \lambda^{2 \nu_1 + j - 1} e^{- \lambda}.
\end{equation}
Regarding this as an initial condition in (\ref{f5}) allows $g_j^{(M)}$ to be expressed as a particular
Meijer G-function. Recall that the Meijer G-function is defined by
\begin{equation}
\MeijerG{m}{n}{p}{q}{a_1,\ldots,a_p}{b_1,\ldots,b_q }{z}=\frac{1}{2\pi i}\int_\gamma
\frac{\prod_{j=1}^m\Gamma(b_j+u)\prod_{j=1}^n\Gamma(1-a_j-u)}
{\prod_{j=m+1}^q\Gamma(1-b_j-u)\prod_{j=n+1}^p\Gamma(a_j+u)}z^{-u}
\, du,
\end{equation}
where $\gamma$ is an appropriate contour relating to the validity of the inverse Mellin transform formula. Most of the often encountered special functions can be viewed as special cases of the Meijer G-functions.

\begin{lemma}\label{lemma-g-ele}
Let $g_j^{(1)}$ be given by (\ref{f6}) and define $g_j^{(s)}$ for $s=2,3,\dots$ by the recurrence (\ref{f5}).
We have
\begin{equation}\label{f7}
g_j^{(M)}(\lambda) =
\MeijerG{M}{0}{0}{M}{-}{2 \nu_M, \dots, 2 \nu_2, 2\nu_1 + j }{\lambda}.
\end{equation}
\end{lemma}

\begin{proof}
For general $\rho$, we have
\[
x^\rho e^{-x} =
\MeijerG{1}{0}{0}{1}{\underline{\hspace{0.2cm}}}{\rho }{x}.
\]
Hence, from (\ref{f5}) and (\ref{f6})
\begin{align*}
g_j^{(2)}(\lambda) & = \int_0^\infty
\MeijerG{1}{0}{0}{1}{-}{2 \nu_2 - 1}{b}
\MeijerG{1}{0}{0}{1}{-}{2 \nu_1 + j }{\frac{\lambda}{b}} \, db \\
& = \int_0^\infty
\MeijerG{1}{0}{0}{1}{-}{2 \nu_2 - 1}{b}
\MeijerG{0}{1}{1}{0}{1 - (2 \nu_1 + j)}{-}{\frac{b}{\lambda}} \, db,
\end{align*}
where the second equality follows from a general formula relating the Meijer G-function with argument
$1/x$ to another Meijer G-function with argument $x$; see e.g.~\cite{Lu69}.  A fundamental property of
Meijer G-functions is that they are closed under convolutions of this type, and (\ref{f7}) results with
$M=2$. The derivation for general $M$ proceeds similarly, making use too of induction.
\end{proof}

So by Lemma~\ref{lemma-g-ele}, we are considering a polynomial ensemble with eigenvalue PDF given by
\begin{equation}
\frac1{Z_{m,M}}\Delta_m^\even(\lambda)\det[g_{j-1}^{(M)}(\lambda_i)]_{i,j=1}^m.
\end{equation}
The corresponding biorthogonal functions are functions $p_n(x)\in\operatorname{span}\{1,x^2,\ldots,x^{2n}\}$ and $\phi_n(x)\in\operatorname{span}\{g^{(M)}_0,g^{(M)}_1,\ldots,g^{(M)}_n\}$ such that
\begin{equation}\label{def:pq}
\int_0^\infty p_n(x)\phi_m(x) \, dx=\delta_{m,n}.
\end{equation}
To specify them it is illuminating to first calculate the bimoments. We have
\begin{equation}\label{bimoment-ele}
B_{k\ell}^{(M)}:=\int_0^\infty x^{2k}g_\ell^{(M)}(x)\, dx=\Gamma(2\nu_1+2k+\ell+1)\prod_{i=2}^M\Gamma(2\nu_i+2k+1),
\end{equation}
while the determinant of the bimoment matrix yields
\begin{equation}
D_m^{(M)}:=\det[B_{k\ell}^{(M)}]_{k,\ell=0}^m=\prod_{k=0}^m2^kk!\prod_{i=1}^M\Gamma(2\nu_i+2k+1),
\end{equation}
and more importantly
\begin{equation}\label{det-frac}
\frac{D_m^{(M)}}{D_m^{(1)}}=\prod_{k=0}^m\prod_{i=2}^M\Gamma(2\nu_i+2k+1).
\end{equation}

Now, we are ready to show
\begin{proposition}
The biorthogonal polynomial $p_n(x)$ in \eqref{def:pq} is given by
\begin{equation}\label{polynomial1-sum}
p_n(x)=\sum_{k=0}^n(-1)^{n-k}\binom{n}{k}\prod_{i=1}^M\frac{\Gamma(2\nu_i+2n+1)}{\Gamma(2\nu_i+2k+1)}x^{2k}.
\end{equation}
\end{proposition}

\begin{proof}
By the general theory of biorthogonal functions~\cite{DF06b}, we know that
\begin{equation}
p_n(x)=\frac1{D_{n-1}^{(M)}}\det\big[B_{k\ell}^{(M)}\big\vert x^{2k} \big]^{k=0,\ldots,n}_{\ell=0,\ldots,n-1}.
\end{equation}
Now, using the structure of the bimoments~\eqref{bimoment-ele} and~\eqref{det-frac}, it follows that
\begin{equation}\label{polynomial-step-proof}
p_n(x)=\prod_{\ell=2}^M\Gamma(2\nu_\ell+2n+1)
\frac1{D_{n-1}^{(1)}}\det\bigg[B_{k\ell}^{(1)}\bigg\vert
\frac{x^{2k}}{\prod_{i=2}^M\Gamma(2\nu_i+2k+1)} \bigg]^{k=0,\ldots,n}_{\ell=0,\ldots,n-1}.
\end{equation}
Thus, the structure for the biorthogonal polynomial for general $M$ is linked to the case $M=1$. Moreover, the $M=1$ case is the so-called Laguerre Muttalib--Borodin ensemble for which the biorthogonal polynomials are known. In particularly, we know from~\cite{Kon67,Zh15,FI16} that
\begin{equation}
\frac1{D_{n-1}^{(1)}}\det\big[B_{k\ell}^{(1)}\big\vert
x^{2k} \big]^{k=0,\ldots,n}_{\ell=0,\ldots,n-1}
=\sum_{k=0}^n(-1)^{n-k}\binom{n}{k}\frac{\Gamma(2\nu_1+2n+1)}{\Gamma(2\nu_1+2k+1)}x^{2k}.
\end{equation}
We can use this expression to evaluate the determinant in~\eqref{polynomial-step-proof}, which completes the proof.
\end{proof}

There is an alternative representation for the polynomial $p_n$.
For this, we need  the generalized hypergeometric function
 \begin{equation}\label{def:hypergeo}
 {\; }_p F_q \left({a_1,\ldots, a_p \atop b_1,\ldots,b_q} \Big{|} z \right)=\sum_{k=0}^\infty \frac{(a_1)_k\cdots (a_p)_k}{(b_1)_k \cdots (b_q)_k}\frac{z^k}{k!}
\end{equation}
   with
\begin{equation}\label{eq:pochammer}
(a)_k=\frac{\Gamma(a+k)}{\Gamma(a)}=a(a+1)\cdots(a+k-1)
\end{equation}
being the Pochhammer symbol.

\begin{corollary}\label{cor:pn}
The biorthogonal polynomial~\eqref{polynomial1-sum} has the generalized hypergeometric function and Meijer G-function representations given by
\begin{align}
p_n(x)&=(-1)^n\prod_{i=1}^M\frac{\Gamma(2\nu_i+2n+1)}{\Gamma(2\nu_i+1)}
\hypergeometric{1}{2M}{-n}{1+\nu_1,\nu_1+\tfrac12,\ldots,1+\nu_M,\nu_M+\tfrac12}{\frac{x^2}{2^{2M}}}
\nonumber
\\
&=(-1)^nh_n\MeijerG{1}{0}{1}{2M+1}{n+1}{0,-\nu_1,\tfrac12-\nu_1,\ldots,-\nu_M,\tfrac12-\nu_M}{\frac{x^2}{2^{2M}}}
\label{polynomial1-meijer}
\end{align}
with $$h_n=n!\prod_{i=1}^M\left(2^{2n}\Gamma(\nu_i+n+1)\Gamma(\nu_i+n+\tfrac12)\right).$$
\end{corollary}

\begin{proof}
The expression~\eqref{polynomial1-sum} is immediately recognized as a generalized hypergeometric sum after using the duplication formula for the gamma functions, hence the first equality in \eqref{polynomial1-meijer} follows. The second equality follows from the relationship between the generalized hypergeometric function and the Meijer G-function,
\begin{equation*}
\hypergeometric{1}{q}{-n}{b_1,\ldots,b_q}{z}=n!\prod_{i=1}^q\Gamma(b_i) \,
\MeijerG{1}{0}{1}{q+1}{n+1}{0,1-b_1,\ldots,1-b_q}{z}.
\end{equation*}
\end{proof}

The biorthogonal function $\phi_n(x)$ can also be expressed as a Meijer G-function.
\begin{proposition}\label{prop:phin}
The function $\phi_n(x)$ in \eqref{def:pq} is given by
\begin{equation}\label{polynomial2}
\phi_n(x)=\frac{(-1)^nx}{2^{2M-1}h_n}
\MeijerG{2M}{1}{1}{2M+1}{-n}{\nu_1,\nu_1-\tfrac12,\ldots,\nu_M,\nu_M-\tfrac12,0}{\frac{x^2}{2^{2M}}}.
\end{equation}
\end{proposition}

\begin{proof}
We know from the work of Akemann et al.~\cite{AIK13} (see also~\cite{Ip15}) that the biorthogonal relation
\begin{equation}
\int_0^\infty(-1)^nh_n\MeijerG{1}{0}{1}{2M+1}{n+1}{0,-\nu_1,\tfrac12-\nu_1,\ldots,-\nu_M,\tfrac12-\nu_M}{x}q_k(x)dx
=\delta_{n,k}
\end{equation}
is satisfied if
\begin{equation}
q_n(x)=\frac{(-1)^n}{h_n}\MeijerG{2M}{1}{1}{2M+1}{-n}{\nu_1,\nu_1-\tfrac12,\ldots,\nu_M,\nu_M-\tfrac12,0}{x}.
\end{equation}
Thus, by comparison with~\eqref{polynomial1-meijer} we know that the functions~\eqref{polynomial2} are biorthogonal respect to the polynomials~\eqref{polynomial1-sum}. Furthermore, it follows from this comparison that $\phi_n$ belongs to the linear span of $g^{(M)}_0, \ldots, g^{(M)}_n$.
\end{proof}

We conclude this subsection with the proof of Corollary \ref{cor1even}.

\paragraph{Proof of Corollary \ref{cor1even}}
Let $X_j$, $j=1,\ldots,M$ be independent complex matrices of size $(n+\nu_j)\times(n+\nu_{j-1})$ with $\nu_0=0$ and $\nu_j\geq 0$. Assume that the matrices $\{X_j\}$ are distributed by identical, independent Gaussians, i.e., with PDF proportional to
$$e^{-\textrm{Tr}X_j^\dagger X_j},$$
and consider the the product
\begin{equation} \label{Ym}
Y_M = X_M X_{M-1} \cdots X_1.
\end{equation}

Denote by $\{x_j\}_{j=1}^n$ the eigenvalues of $Y_M^\dagger Y_M$. According to \cite{AIK13}, the
PDF of $\{x_j\}$ is given by
\begin{equation} \label{jpdf}
  \frac{1}{\mathcal{Z}_n}  \Delta(x)
        \det \left[ w_{k-1}(x_j) \right]_{j,k=1}^n,
\end{equation}
where $\mathcal{Z}_n = n!\prod_{i=1}^{n}\prod_{j=0}^M \Gamma(i+\nu_j)$ and
\begin{equation} \label{wk}
    w_k(x) = \MeijerG{M}{0}{0}{M}{-}{\nu_M, \nu_{M-1},  \ldots, \nu_2, \nu_1 +
    k}{x}
    \end{equation}
Let $P_n(x)$ and $Q_n(x)$ be the associated biorthogonal functions. That is,
$P_n$ is a monic polynomial of degree $n$ and $Q_n$ is a linear combination of $w_0, \ldots, w_n$, uniquely defined by the orthogonality
\begin{equation} \label{PkQkbio}
    \int_0^{\infty} P_j(x) Q_k(x) \, dx = \delta_{j,k}.
    \end{equation}
By \cite{AIK13}, we have the following explicit formulas of $P_n$ and $Q_n$ in terms of Meijer G-functions:
\begin{equation} \label{PnMeijerG}
  P_n(x)  = (-1)^n
    \prod_{j=0}^M\Gamma(n+\nu_j+1)\MeijerG{0}{1}{1}{M+1}{n+1}{0,-\nu_M, \ldots, -\nu_1}{x}
   \end{equation}
   and
    \begin{align} \label{QkMeijerG}
  Q_n(x) = \frac{(-1)^n}{\prod_{j=0}^M\Gamma(n+\nu_j+1)} \MeijerG{M}{1}{1}{M+1}{-n}{\nu_M, \nu_{M-1}, \ldots, \nu_1,0}{x}.
    \end{align}

By Corollary \ref{cor:pn} and Proposition \ref{prop:phin}, it is easily seen that if we make a change of variables $x^2/2^{2M}\mapsto x$, then we find exactly the same biorthogonal system for two different random products, and the sought result  follows.
 \qed


\subsection{Initial matrix as Hermitian anti-symmetric Gaussian random matrix}\label{S4.4}

In this subsection, we identify another case of~\eqref{f3} which also turns out to be closely related to products of complex Gaussian matrices. In this case the initial matrix is chosen as $$B_0=iA$$
with $A$ a $2m\times 2m$ real anti-symmetric standard Gaussian random matrix.
The eigenvalue PDF of $B_0$ in this setting is given by the functional form (\ref{f1}) with
\begin{equation}\label{f8}
Z_m = \pi^{m/2} 2^{-m^2} \prod_{j=0}^{m-1} (2j)!, \qquad
g_{j-1}(b) = b^{2(j-1)} e^{-b^2};
\end{equation}
see e.g.~\cite[Eq.~(2.18)]{FN09}.
With $g_j^{(0)}(b)$ given by this $g_j$, it so happens that the recurrence (\ref{f5}) has been
encountered in the recent study \cite{FIL17} of the eigenvalues of the random matrix product
\[
G_M^\dagger \cdots G_1^\dagger H G_1 \cdots G_M,
\]
where each $G_j$ is a standard complex Gaussian matrix of size
$(\nu_{j-1} + N ) \times (\nu_j + N)$ and $H$ is a member of the GUE of size $N\times N$. The
only difference is that $\{2 \nu_s \}_{s=0,1,\dots}$ in (\ref{f5}) and $2(j-1)$ in (\ref{f8}) need to be replaced by $\{ \nu_s \}_{s=0,1,\dots}$  and $(j-1)$ respectively to obtain a recurrence of \cite{FIL17}.
Moreover, the recurrence was solved in this work, first in terms of a contour integral
\cite[Eq.~(2.10)]{FIL17}, and this in turn was recognised as a particular Meijer G-function
\cite[Eq.~(2.14)]{FIL17}, telling us in the present setting that
\begin{equation}\label{f9}
g_j^{(M)}(\lambda) =
\prod_{\ell=1}^M \frac{2^{\nu_\ell - 1}}{\sqrt{\pi}}
\MeijerG{2M+1}{0}{0}{2M+1}{-}{\nu_1,\nu_1 + 1/2,\dots, \nu_M, \nu_M + 1/2,  j }{\frac{\lambda^2}{4^M}}.
\end{equation}
Hence with $g^{(M)}_j(\lambda)$ so defined, the eigenvalue PDF for the product matrix~\eqref{f3} with initial matrix $B_0$ a Hermitian anti-symmteric matrix is
\begin{equation}\label{joint-dist-even}
\frac1{Z_m}\Delta_m^\even (\lambda)\det[g_j^{(M)}(\lambda_i)]_{i,j=1}^m
\end{equation}
with $0\leq\lambda_1\leq\cdots\leq\lambda_m$. This structure looks very similar to that for the Hermitised product studied in~\cite{FIL17}, with the important difference that the eigenvalue PDF only depends on the square of the eigenvalues $\{\lambda_i\}$. This latter fact is closely related to the reflection symmetry of the spectrum about the origin. We exploit this symmetry by making a change of variables $x_i=\lambda_i^2/4^M$ ($i=1,\dots,m$) in~\eqref{joint-dist-even}, which yields
\begin{equation}\label{AIK-reappears}
\frac1{Z_{m,M}}\Delta_m (x)
\det\Big[\MeijerG{2M+1}{0}{0}{2M+1}{-}{\nu_1-\tfrac12,\nu_1,\ldots,\nu_M-\tfrac12,\nu_M,j-\tfrac12}{x}\Big]_{i,j=1}^m,
\end{equation}
where $Z_{m,M}$ is a modified normalisation constant.

Like Corollary \ref{cor1even}, we can make the crucial observation that the functional form of~\eqref{AIK-reappears} is exactly the same as the model studied by Akemann et al.~\cite{AIK13} for the eigenvalues for the product ensemble
\begin{equation}\label{product-odd}
G_{2M+1}^\dagger\cdots G_1^\dagger G_1\cdots G_{2M+1}.
\end{equation}
Here, each $G_i$ is an $m\times m$ complex Gaussian matrix with PDF proportional to
\begin{equation}\label{4.39}
\det (G_i^\dagger G_i)^{\nu_i-1/2} e^{-\Tr G_i^\dagger G_i}
\qquad \text{and} \qquad
\det (G_i^\dagger G_i)^{\nu_i} e^{-\Tr G_i^\dagger G_i}
\end{equation}
for $i=2j-1$ ($j=1,\ldots,m+1$) with $\nu_{2M+1}=0$ and $i=2j$ ($j=1,\ldots,m$), respectively.
We note that the product~\eqref{product-odd} consists of an odd number of factors, while the product~\eqref{product-even} consists of an even number of factors. Nonetheless, the conclusion is the same. We can read off the form of the correlation kernel as well as its scaling behaviour from the preexisting literature.

\section{Special cases for odd dimensions}
\label{sec5}

\subsection{Rank two perturbation}

The analogue of (\ref{3.13}) in the odd dimensional case is
\begin{equation}\label{3.13odd}
\operatorname{diag}\Big ((a_1,\ldots,a_n)\otimes
\begin{bmatrix}
0 & i \\
-i & 0
\end{bmatrix},0 \Big )
+ X_{(2n+1)\times 2}
\begin{bmatrix}
0 & ib_1 \\
-ib_1 & 0
\end{bmatrix}(X^T)_{2\times (2n+1)}.
\end{equation}
We can deduce the corresponding eigenvalue PDF by specialising Theorem
\ref{prop4odd}, analogous to the derivation of Corollary \ref{cor4.1} from Theorem \ref{prop4}.

\begin{corollary}
In the case $t=0$, $l=2$, the random matrix \eqref{3.4odd} has the same eigenvalue PDF as \eqref{3.13odd}. This eigenvalue PDF is given by \eqref{3.5odd} with
\begin{equation}\label{3.14odd}
g_i(\lambda)=\frac{1}{2b_1}\left(e^{-(|\lambda-a_i|/b_1}-e^{-\lvert\lambda+a_i\rvert/b_1}\right).
\end{equation}
\end{corollary}
\begin{proof}
According to \eqref{3.6odd}, in the case $t=0$, and with $B$ given as implied by (\ref{3.13odd})
\begin{equation*}
g_i(\lambda)=\frac{2}{\pi}\int_0^\infty \frac{\sin (a_ic) \sin (\lambda c)}{1+b_1^2c^2}\,dc.
\end{equation*}
Evaluating this integral gives \eqref{3.14odd}.
\end{proof}

\begin{remark} With $g_k$ as given by  \eqref{3.6odd}, and with $b_1 = 1$ for convenience,
a result of  Defosseux \cite{De10}  gives
\begin{equation}\label{X0odd}
\det[g_i(\lambda_j)]_{i,j=1}^n= \\
\frac{1}{2}\int
\chi({z\prec\lambda})\chi({z\prec a}) e^{-\sum_{j=1}^n(a_j+\lambda_j - 2 z_j)}\,dz_1 ...dz_{n};
\end{equation}
cf.~(\ref{X0}). As with (\ref{X0}), this is seen to vanish unless $a \prec \lambda$ or $\lambda \prec a$.
\end{remark}

\subsection{Hermitised products}
We now consider the matrix~\eqref{3.4odd} with $t=0$ and $A=0$. As in the derivation of \eqref{3.15}, by taking $a_i \to 0$, it is readily seen that the PDF in Theorem \ref{prop4odd} reduces to
\begin{equation}\label{3.15odd}
\prod_{l=1}^{n-1}\frac{1}{(2l+1)!}{\Delta}_n^\odd (\lambda)
\det\bigg[\frac{2}{\pi}\int_0^\infty \frac{(-1)^{k-1}c^{2k-1} \sin (\lambda_j c)}{\det(\mathbb{I}+icB)}\,dc\bigg]_{j,k=1}^n.
\end{equation}
Since $\det(\mathbb{I}+icB)=\prod_{l=1}^m (1+c^2b_l^2)$ for $l=2m+1$, we again assume that $m\geq n$ to ensure that all the integrals in \eqref{3.15odd} are well-defined. We then have the following odd dimensional analogue of Proposition \ref{prop5}.

\begin{proposition}\label{prop5odd}
In the case $l=2n+1$, the eigenvalue PDF \eqref{3.15odd} corresponding to $t=0$, $A=0$ in Theorem \ref{prop4odd} and thus of the random matrix $XBX^T$ reduces to
\begin{equation}\label{3.16odd}
\prod_{l=1}^{n-1}\frac{1}{(2l+1)!} \prod_{l=1}^{n} \frac{\lambda_l}{b_l}
\frac{{\Delta}_n^\even (\lambda)}{{\Delta}_n^\odd (b)}\det[e^{-\lambda_j/b_k}]_{j,k=1}^n.
\end{equation}
\end{proposition}
\begin{proof}
For the integral in \eqref{3.15odd}, we have, with the aid of the residue theorem,
\begin{align}
&\frac{2}{\pi}\int_0^\infty \frac{(-1)^{k-1}c^{2k-1} \sin (\lambda_j c)}{\det(\mathbb{I}+icB)}\,dc
=\frac{1}{\pi}\int_{-\infty}^\infty \frac{(-1)^{k-1}c^{2k-1} \sin (\lambda_j c)}{\det(\mathbb{I}+icB)}\,dc
\nonumber
\\
& =-\frac{i}{\pi}\int_{-\infty}^\infty \frac{(-1)^{k-1}c^{2k-1} e^{i \lambda_j c}}{\prod_{l=1}^m (1+c^2b_l^2)}\,dc
= \sum_{\theta=1}^m \frac{e^{-\lambda_j/b_\theta}b_\theta^{2(m-k-1)}}{\prod_{l=1,l\neq \theta}^m(b_\theta^2-b_l^2)}.
\end{align}

Inserting the above formula into \eqref{3.15odd} shows that in the special case $l=2n+1$, the determinant therein can be written
\begin{align*}
& \frac{1}{({\Delta}_n^\odd (b))^2}\sum_{k_1,...,k_n=1}^n \prod_{j=1}^n e^{-\lambda_j/b_{k_j}}\det[b_{k_j}^{2(i-1)}]_{i,j=1}^n
\nonumber
\\
& = \frac{\det[e^{-\lambda_j/b_k}]_{j,k=1}^n \det[b_{k_j}^{2(i-1)}]_{i,j=1}^n}{\prod_{l=1}^nb_l{\Delta}_n^\odd (b){\Delta}_n^\even(b)}=\frac{\det[e^{-\lambda_j/b_k}]_{j,k=1}^n}{\prod_{l=1}^nb_l{\Delta}_n^\odd (b)},
\end{align*}
following the same strategy in the derivation of the even dimensional case. This, together with \eqref{3.15odd}, gives us
\eqref{3.16odd}.
\end{proof}

Hence, by considering the limiting case such that some of $\lambda_i$ and $b_i$ in \eqref{3.16odd} tend to zero, the following odd dimensional analogue of Corollary \ref{cor4} is immediate.

\begin{corollary}\label{cor4odd} (Defosseux \cite{De10}).
Let $X$ be a $(2n+1)\times (2m+1)$ ($n\geq m$) standard real Gaussian matrix. With $b_j>0$ ($j=1,\dots,m$), let
\begin{equation}\label{Bp}
B=\operatorname{diag}\Big(
\big[\begin{smallmatrix} 0 & ib_1 \\ -ib_1 & 0 \end{smallmatrix}\big],
\ldots,
\big[\begin{smallmatrix} 0 & ib_m \\ -ib_m & 0 \end{smallmatrix}\big],0 \Big).
\end{equation}
The eigenvalue PDF of $XBX^T$, or equivalently of $\tilde{X} \tilde{B} \tilde{X}^T$, where
$\tilde{X}$ is a $(2n+1) \times (2n+1)$ standard Gaussian matrix and $\tilde{B}$ is the 
block diagonal matrix
\[
\tilde{B} =
\operatorname{diag}\Big(
\big[\begin{smallmatrix} 0 & ib_1 \\ -ib_1 & 0 \end{smallmatrix}\big],
\ldots,
\big[\begin{smallmatrix} 0 & ib_m \\ -ib_m & 0 \end{smallmatrix}\big],
\underbrace{%
\big[\begin{smallmatrix} 0 & 0 \\ 0 & 0\vphantom{b_m} \end{smallmatrix}\big]
\ldots,
\big[\begin{smallmatrix} 0 & 0 \\ 0 & 0\vphantom{b_m} \end{smallmatrix}\big]%
}_{n-m\textup{ times}},0
\Big),
\]
is equal to
\begin{equation}\label{3.18odd}
\prod_{l=n-m}^{n-1}\frac{1}{(2l+1)!}\frac{{\Delta}_m^\even (\lambda)}{{\Delta}_m^\odd (b)} \prod_{l=1}^{m}(\lambda_l/b_l)^{2(n-m)+1} \det[e^{-\lambda_j/b_k}]_{j,k=1}^{m}.
\end{equation}
\end{corollary}

Finally, by assuming that the eigenvalue PDF of $B =: B_0$ in the above corollary is of the
polynomial ensemble form
\begin{equation}\label{f1odd}
\frac{1}{Z_m} \Delta_m^\even (b_j )
\det [ g_{j-1}(b_k) ]_{j,k=1}^m,
\end{equation}
we can extend the matrix product $X B X^T$ to involve an arbitrary number of Gaussian matrices, as in the even dimensional case.


\begin{corollary}
Let us consider a sequence of matrix products of the form
\begin{equation}\label{f3odd}
B_M = X_M \cdots X_1 B_0 X_1^T \cdots X_M^T,
\end{equation}
where $B_0$ is as stated in Corollary \ref{cor4odd} with $\{b_j\}_{j=1}^m$ distributed according to \eqref{f1odd}, and each $X_j$ $(j=1,2,\ldots,M)$ is a $(2(m+\nu_j)+1) \times (2(m+\nu_{j-1})+1)$ standard real Gaussian matrix with $\nu_j\geq \nu_{j-1}$ and $\nu_0=0$. Then, the PDF for the eigenvalues of $B_M$ is also of the polynomial ensemble type and is given by
\begin{equation}\label{f4odd}
\frac{1}{Z_m} \prod_{j=1}^M\prod_{l=\nu_j}^{m+\nu_j-1}  \frac{1}{(2l+1)!}\Delta^\even_m (\lambda)
\det \Big [ g_{j-1}^{(M)}(\lambda_k) \Big ]_{j,k=1}^m,
\end{equation}
where $g_j^{(M)}$ is defined recursively according to
\begin{equation}\label{f5odd}
g_j^{(s)}(\lambda) = \int_0^\infty e^{-b} b^{2 \nu_s } g_j^{(s-1)} \Big ( \frac{\lambda}{b} \Big ) \, d b
\end{equation}
with $g_j^{(0)}(\lambda) := g_j(\lambda)$.
\end{corollary}


\subsection{Elementary anti-symmetric matrix for $B_0$}
Suppose $b_1 = \cdots = b_m = 1$ in (\ref{Bp}), that is,
\begin{equation}\label{def:B0odd}
B_0=\Big ( \Big ( \mathbb I_m \otimes \begin{bmatrix} 0 & i \\ -i & 0 \end{bmatrix} \Big )
\oplus [0] \Big ) .
\end{equation}
By an appropriate limiting procedure, formula
(\ref{3.18odd}) reduces to
$$
{1 \over 2^{m(m-1)/2}}
\prod_{l=1}^{m-1} {1 \over l!} \prod_{l=n-m}^{n-1} {1 \over (2l+1)!} \prod_{l=1}^m
\lambda_l^{2(n-m)+1} e^{- \lambda_l}
\prod_{1 \le j < k \le m} ( \lambda_k - \lambda_j) (\lambda_k^2 - \lambda_j^2).
$$
This is functionally identical to (\ref{3.19}), except that the exponent in the Laguerre weight has
been shifted up by 1. If we now consider the product (\ref{f3odd}) with $B_0$ given by \eqref{def:B0odd}, it follows that $g_{j-1}^{(1)}(\lambda)$ is given by (\ref{f6}) with the exponent increased by 1. Repeating the calculations of \S \ref{sec:ele}, using now the recursion (\ref{f5odd}), we arrive at the following result, which is a minor
variant of the previous findings.

\begin{proposition}\label{prop5.6}
Consider the product (\ref{f3odd}) with $B_0$ given in \eqref{def:B0odd}. With $\{ \lambda_j \}_{j=1}^m$ denoting the positive eigenvalues, the variables
$x_j = \lambda_j^2/2^{2M}$ are distributed as for the eigenvalues of the product ensemble
(\ref{product-even}), where $\{G_i\}$ are $m \times m$ complex random matrices with PDF
proportional to (\ref{4.33}), but with $\nu_{i} \mapsto \nu_{i} + \frac{1}{2}$.
\end{proposition}

We remark that product matrices in Proposition \ref{prop5.6} can be written $G^\dagger \Sigma
G$, where $G$ has integer parameter $\nu_{2M}$ and $\Sigma = G_{2M-1}^\dagger \cdots G_1^\dagger
G_1 \cdots G_{2M-1}$ plays the role of a (random) covariance matrix.

\subsection{Hermitian Gaussian anti-symmetric matrix for $B_0$}
We know from \S \ref{S4.4} that for the even dimensional case, choosing $B_0$ in (\ref{f3}) as an
Hermitian Gaussian anti-symmetric matrix  gives an eigenvalue PDF which can be
related to the eigenvalue PDF for a certain product of complex Gaussian matrices. The same general feature carries over to the odd dimensional case.

The first point to note is that in the odd-dimensional case, with $B_0 = i A$ and $A$ a
$(2m+1) \times (2m+1)$ real anti-symmetric standard Gaussian random matrix, the positive
eigenvalue PDF of $B_0$ is given by the functional (\ref{f1}) with $g_j(b) = b^{2(j-1)+1} e^{-b^2}$.
This is to be substituted for $g_j^{(0)}(b)$ in (\ref{f5odd}), which then is to be computed recursively.
Comparison of (\ref{f5odd}) with (\ref{f5}), the latter specified by choosing $g_j^{(0)}(b)$ equal to the
value of $g_j(b)$ in (\ref{f8}), we see that $g_j^{(s)}(\lambda)$ as given by (\ref{f5odd}) is equal to
$\lambda$ times $g_j^{(s)}(\lambda)$ as given by (\ref{f5}). Thus, after the change of variables
$x_i = \lambda_i^2/4^M$, the joint eigenvalue PDF is given by (\ref{AIK-reappears}) with all the
indices in the bottom row of $G_{0,2M+1}^{2M+1,0}$ increased by $\frac{1}{2}$. Comparison
now with the results in \cite{AIK13}, a companion to the results of Proposition \ref{prop5.6}
is obtained.

\begin{proposition}\label{prop5.7}
Consider the product (\ref{f3odd}) with $B_0 = i A$ and $A$ a
$(2m+1) \times (2m+1)$ real anti-symmetric standard Gaussian random matrix.
With $\{ \lambda_j \}_{j=1}^m$ denoting the positive eigenvalues, the variables
$x_j = \lambda_j^2/2^{2M}$ are distributed as for the eigenvalues of the product ensemble
(\ref{product-odd}), where $\{G_i\}$ are $m \times m$ complex random matrices with PDF
proportional to (\ref{4.39}), but with $\nu_{i} \mapsto \nu_{i} + \frac{1}{2}$.
\end{proposition}

\section*{Acknowledgements}
We acknowledge support by the Australian Research Council through grant DP170102028 (PJF), the ARC Centre of Excellence for Mathematical and Statistical Frontiers (PJF, JRI), by
 the Natural Science Foundation of China \# 11771417, the Youth Innovation Promotion Association CAS  \#2017491, the Fundamental Research Funds for the Central Universities  \#WK0010450002 and Anhui Provincial Natural Science Foundation \#1708085QA03 (DZL), and by the National Natural Science Foundation of China \#11501120, the Program for Professor of Special Appointment (Eastern Scholar) at Shanghai Institutions of Higher Learning, Grant \#EZH1411513 from Fudan University (LZ).
Mario Kieburg is to be thanked for discussions which motivated \S \ref{sect3.5}.


\providecommand{\bysame}{\leavevmode\hbox to3em{\hrulefill}\thinspace}
\providecommand{\MR}{\relax\ifhmode\unskip\space\fi MR }
\providecommand{\MRhref}[2]{%
  \href{http://www.ams.org/mathscinet-getitem?mr=#1}{#2}
}
\providecommand{\href}[2]{#2}

\end{document}